\documentclass[11pt]{article}

\usepackage{color}
\usepackage{subfigure}
\usepackage{booktabs}
\usepackage{graphicx}
\usepackage{fullpage}

\usepackage{bm}
\usepackage{amsthm}
\usepackage{amsmath}
\usepackage{amsfonts}
\usepackage{amssymb}

\usepackage[colorlinks,citecolor=blue,urlcolor=blue]{hyperref}
\usepackage[authoryear, round]{natbib}

\newcommand{\R}{\mathbb{R}}
\newcommand{\N}{\mathbb{N}}
\newcommand{\p}{\phantom{)}}
\newcommand{\ppp}{\phantom{.00}}

\renewcommand{\Pr}{\mathbb{P}}
\newcommand{\E}{\mathbb{E}}

\theoremstyle{plain}

\newtheorem{theorem}{Theorem}[section]
\newtheorem{lemma}[theorem]{Lemma}
\newtheorem{proposition}[theorem]{Proposition}

\theoremstyle{remark}
\newtheorem{definition}[theorem]{Definition}
\newtheorem{example}{Example}

\begin{document}

\title{Likelihood-based inference for modelling packet transit from thinned flow summaries}

\author{
P.A. Rahman\footnote{School of Mathematics and Statistics, \textsc{unsw} Sydney, Australia}\;\footnote{Communicating author: \texttt{P.A.Rahman@unsw.edu.au}}\;,
B.Beranger$^*$,
M. Roughan\footnote{School of Mathematical Sciences, University of Adelaide, Australia}\;,
and\; S.A. Sisson$^*$
}

\date{}

\maketitle

\begin{abstract}
The substantial growth of network traffic speed and volume presents practical challenges to network data analysis.
Packet thinning and flow aggregation protocols such as NetFlow
reduce the size of datasets by providing structured data summaries, but conversely this impedes statistical inference. 
Methods which aim to model patterns of traffic propagation typically do not account for the packet thinning and summarisation process into the analysis, and are often simplistic, e.g.~method-of-moments. As a result, they can be of limited practical use.

We introduce a likelihood-based analysis which fully incorporates packet thinning and NetFlow summarisation into the analysis.
As a result, inferences can be made for models on the level of individual packets while only observing thinned flow summary information.
We establish consistency of the resulting maximum likelihood estimator, derive bounds on the volume of traffic which should be observed to achieve required levels of estimator accuracy, and identify an ideal family of models.
The robust performance of the estimator is examined through
simulated analyses and an application on a publicly available trace dataset containing over 36m packets over a 1 minute period.
\end{abstract}


\section{Introduction}

Network traffic volumes and speeds have grown exponentially since the inception of the internet \citep{Cisco19}.
Accordingly,  accounting compromises, such as \emph{flow-level aggregation} and \emph{packet thinning}, are typically employed \citep{Hofstede14}. 

In flow-level aggregation (flows are sets of packets with particular grouping identifiers), broad flow characteristics are recorded instead of individual packet information.
These typically include the flow size (number of packets), start time, and duration.
Commonly used protocols include \textsc{ipfix} and \emph{NetFlow} (here we use NetFlow as a general term to denote flow summarisation).
NetFlows reduce the volume of stored information as each flow, of potentially millions of packets, is summarised into typically eight numbers \citep{Hofstede14}.

Packet thinning is a parallel strategy whereby only selected packets are recorded, either randomly or deterministically, as they pass through the network. 
Thinning techniques include flow sampling \citep{Hohn06}, adaptive sampling, and simple random sampling \citep{Jurga07}.
We consider the most elementary probabilistic design, Bernoulli sampling, in which individual packets are independently recorded with some constant known probability $q$ \citep{Hofstede14, Hohn06, Jurga07}.
In practice, sampling rates can be tuned to match the density and speed of traffic in the network, so that, for example, a network with speeds of 1$\text{Gb}/s$ would typically tune the sampling rate to be $q = 0.001$.
Flow-level aggregation can also be applied to packet thinned traffic.

Such data retention strategies need to be considered explicitly when analysing the summarised data, since bias, and other errors, may otherwise arise.
In traffic classification, for example, basic analysis on thinned traffic will over-represent large media applications such as video streaming since these flows are significantly larger.
Smaller but more numerous applications, such as e-mails, will conversely be under-represented.

Many of the more sophisticated network analysis techniques address either flow-level aggregation \citep{Bin08} or packet thinning \citep{Hohn03, Hohn06, Antunes16}, but fail to jointly consider both.
Analyses which have mutually addressed packet thinning and flow summarisation have assessed network volume and traffic classification \citep{Bin08, Carela10}.
In contrast, we wish to perform parametric inference on patterns of traffic propagation when observing only NetFlows obtained from thinned traffic. 

In this article, we adapt recent results of \citet{Beranger18} and \citet{Zhang19} in Symbolic Data Analysis (\textsc{sda}) to develop a likelihood-based approach for modelling packet-level network traffic data.
The resulting \emph{NetFlow likelihood} incorporates packet thinning and flow aggregation within a generative framework.
As a result, we are able to fit models, and make inference, on packet-level traffic patterns when observing only flow-level summaries (which may have been constructed from packet thinned measurements).

We also make three key contributions within traffic modelling and \textsc{sda} theory.
We first demonstrate that the NetFlow maximum likelihood estimator attains the consistency and asymptotic Normality typical of standard likelihood estimators, which are based on the full dataset. This is the first such result for SDA likelihood-based methods.
We then provide comparative bounds on the loss of information from the aggregation and thinning procedures. 
From this we are able to identify a lower bound on the minimum number of (packet thinned) NetFlows required to produce an estimator which approximates the efficiency of the \textsc{mle} computed on full data.
We also identify a family of packet-level models for which inference on the aggregated and complete data are identical.
Finally, in order to facilitate comparison with existing inference approaches, we introduce an extension of the moments-based estimators of \citet{Hohn03} for NetFlow data.  

While the NetFlow estimator has many desirable properties, it requires higher computational overheads compared to other, more statistically simple approaches, such as the method of moments.
Further, the computation of the estimator under packet thinning requires estimation, or prior knowledge, of the distribution of the flow sizes; although an empirical approximation of this can be obtained relatively easily in practice.

This article is structured as follows. We provide some background  to existing methods of network analysis, the assumed packet transit model, and our framework of analysis (\textsc{sda}) in Section~\ref{sec:related_work}.
A mathematical representation of the NetFlow is presented Section~\ref{sec:a_new_estimator}, which then allows us to define parametric NetFlow likelihoods for complete and thinned traffic.
We provide two of our key contributions, consistency and relative information loss, in Section~\ref{sec:attributes_of_the_estimator}.
Our third contribution, the optimal family of models, is presented in Section~\ref{sec:optimal_model}.
Sections~\ref{sec:simulations} and \ref{sec:real_data} are respectively dedicated to comparative analyses of the NetFlow estimator on synthetic data, and an application to real data. We conclude with a Discussion.


\section{Related work} \label{sec:related_work}
\subsection{Existing methods for traffic analyses}
Significant progress has been made on methods for assessing network volume \citep{Duffield05, Antunes16, Veitch15, Yang07, Chabchoub10, Ribeiro06, Brownlee02} and traffic classification \citep{Kundu07, Carela10, Miao16}.
Inferential methods for analysing patterns of traffic propagation are comparatively less well developed.
Existing approaches, such as series inversion \citep{Hohn06, Antunes16}, wavelets \citep{Hohn03, Stoev05}, empirical distributional estimation \citep{Antunes11}, cluster analysis \citep{Kim19}, time series \citep{You99}, and principal component analysis \citep{Lakhina04} typically fit models indirectly based on empirical characteristics. 
In some cases, the intention is to simply detect deviations from typical traffic behaviour \citep{Stoev05, Bin08, Proto10} or produce elementary network-behaviour statistics \citep{Proto10, Lee10}.
\citet{Hohn03} and \citet{Antunes18} provided simple statistical schemes whereby the parameters of a limited family of parametric models could be identified as a function of the data using method-of-moments estimation.

Methods which are applicable to thinned network traffic, and other adversarial contexts, tend to focus on the simpler analytic challenge of estimating network volume \citep{Duffield05, Ribeiro06, Chabchoub10}.
Techniques for traffic pattern analyses which also account for packet thinning typically fit secondary characteristics, such as moments, and so are limited in their flexibility and inferential use \citep{Hohn03, Hohn06, Antunes11, Antunes16}.
There are some existing methods which analyse NetFlow data, but which do not also account for packet thinning \citep{Hofstede14, Carela10, Proto10}.

The approach we develop here accounts for both packet thinning and NetFlow summarisation when modelling patterns of traffic propagation within the likelihood framework.

\subsection{The Bartlett--Lewis traffic model}
Renewal processes form a natural context for internet traffic and have previously been used extensively in traffic analysis \citep[e.g.][]{Antunes11, Antunes16, Antunes18, Gonzales10, Hohn03, Hohn06, Williamson01}.
However, \citet{Williamson01} argues that simple renewal processes are limited in their ability to jointly model packet-level \emph{burstiness} and the interaction between flows.
\citet{Hohn03} address this concern by suggesting the use of the Bartlett--Lewis renewal process; a \emph{branching renewal process} with applications in many areas including modelling storm patterns, vehicular traffic, and computer failures \citep{Onof93, Bartlett63, Lewis64a, Lewis64b}.
A substantive discussion of the original formulations of the process in \citet{Bartlett63} and \citet{Lewis64a} is provided in \citet{Daley02}.
The Bartlett--Lewis model also has natural extensions into higher dimensions \citep{Bartlett64, Daley02}.

The Bartlett--Lewis process is a sub-class of cluster renewal processes generated by two concurrent processes.
The main and subsidiary processes respectively define the Poissonian arrival of clusters and, conditionally on the arrival of each cluster, the arrival of individual points within the cluster.
The subsidiary processes are finite unidirectional random walks whose origin is the main arrival.
Contextually, the first packet in each flow forms the main process, whilst the subsequent packets form translated simple finite renewal processes.
Observed traffic is then obtained by superimposing the main process of leading packets and all its subsidiary processes of non-leading packets.
In some contexts such as the analysis of clustered vehicular traffic by \citet{Bartlett63}, the main process is not directly observable.

The Bartlett--Lewis process can be equivalently defined as a sequence of independent finite \emph{delayed} renewal processes \citep{Daley02}.
Through this framework we can develop likelihoods for both complete and flow aggregated data.
Packet thinning in the explicit context of Bartlett--Lewis processes has not been studied, but the general results of \citet{Antunes11} can be applied quite naturally.

\subsection{Symbolic data analysis}
Symbolic data analysis is a relatively new field of statistics which models distributions as its fundamental data unit \citep{Billard04}.
Observations in classical statistics are typically points in a Euclidean space, having no internal variation.
However, aggregation of classical data into distributional \emph{symbols} leads to objects with internal variation \citep{Billard03, Billard07}. 
For example, the simplest symbol is the extremal interval, obtained by mapping a finite set of random variables $\bm{X}=(X_1,\ldots,X_n)$ to its extrema $S = \left[X_{(1)}, X_{(n)}\right]$, where $X_{(i)}$ is the $i$-th order statistic of $\bm{X}$.
The remaining $n - 2$ points are then latently distributed within the interval $S$ \citep{Billard03}.

\textsc{sda} methods are designed to analyse data in such distributional forms, in addition to random intervals, including data described by random histograms \citep{Beranger18, Billard03, Billard07, Billard11, Whitaker19a, Whitaker19b} and weighted lists \citep{Billard04, Billard07}.
\textsc{sda} techniques can consequently achieve computational and storage efficiency without sacrificing statistical validity.
Many common statistical procedures have been extended to symbolic data, including regression \citep{Billard02, Billard07, Neto08, Whitaker19b}, likelihood-based inference \citep{Beranger18, Billard11, Brito12, Whitaker19a}, principal component analysis \citep{Brito11, Lauro00}, clustering \citep{Verde04}, and time series analysis \citep{Brito11}.
\textsc{sda} has been used to analyse problems in a broad range of fields including meteorology \citep{Brito12}, finance \citep{Brito12, Beranger18}, medicine \citep{Billard02, Beranger18}, agriculture \citep{Whitaker19b}, ecology \citep{Lin17}, and climatology \citep{Whitaker19a}. 

Here, we follow the likelihood-based approach of \citet{Beranger18} who construct the marginal likelihood for a symbol by accounting for the aggregation function applied to the underlying data and its generative statistical model.
\begin{proposition}[Generative symbolic likelihood \citep{Beranger18}] \label{prop:generative_symbolic_likelihood}
The likelihood function of the symbolic observation $s$ is given by
\begin{equation} \label{eq:generative_symbolic_likelihood}
    \mathcal{L}_{S}{\left(s;\theta,\vartheta\right)} \propto \int_{\mathcal{X}} f_{S\vert\bm{X}}{\left(s;\bm{x},\vartheta\right)}\,g_{\bm{X}}{\left(\bm{x};\theta\right)}\;\mathrm{d}\bm{x}, 
\end{equation}
\noindent where $f_{S\vert\bm{X}}$ is the conditional density of the symbol $S$ given the classical data $\bm{X}$, and $g_{\bm{X}}$ is the joint density (i.e.~classical likelihood function) of $\bm{X}$.
\end{proposition}

The likelihood \eqref{eq:generative_symbolic_likelihood} permits fitting of the underlying model $g_{\bm{X}}$ when only observing the distributional summary $S$ rather than the full dataset $\bm{X}$. 
In the present context, this allows us to fit models for packet-level data when observing only the flow aggregated summaries.

The symbolic likelihood \emph{reduces} to the classical likelihood as the granularity of the aggregation function becomes more fine.
For example, \eqref{eq:generative_symbolic_likelihood} reduces to the standard likelihood $g_{X}(a;\theta)$ when one takes $b\downarrow a$ for the random interval $s = (a, b)$.
In this sense, the symbolic likelihood \eqref{eq:generative_symbolic_likelihood} can be seen as an approximation to the classical likelihood, where the process of summarising $\bm{X}$ to $S$ may induce some loss of information \citep{Beranger18}.
This property similarly holds in our definition of the NetFlow (Definitions~\ref{def:flow_symbol}), which is in essence an extremal interval on unidirectional data. 


\section{A new estimator} \label{sec:a_new_estimator}
\subsection{The NetFlow likelihood} \label{sec:netflow_likelihood}
Classical parametric likelihood-based inference for renewal processes uses the set of inter-renewals as observed data \citep{Karr91}. 
Inference on entire sessions can then be made by collating the inter-renewals from all observed flows.
However, this method is not immediately applicable once the flows have been aggregated into their NetFlow summaries as the individual inter-renewals are then lost.
We address this problem by first observing NetFlows as a particular type of interval-valued random variable.
We then adapt the likelihood of the renewal process to the information contained in the NetFlow symbol using \eqref{eq:generative_symbolic_likelihood}.

Before constructing the NetFlow likelihood, we first define an appropriate symbolic analogue for the NetFlow.

\begin{definition}[NetFlow symbol] \label{def:flow_symbol}
Consider a sequence of inter-renewals $\bm{X} = {\left(X_{1},\ldots,X_{M + 1}\right)}$ of random length $M$.
The NetFlow symbol $S$ of $\bm{X}$ is the image of the aggregation function
\begin{displaymath} 
\begin{aligned}
  &\varphi:\R^{(M + 1)}_{+}\rightarrow\R_{+}\times\R_{+}\times\N \\
  &\varphi{\left(\bm{X}\right)} = {\left(X_{1}, 1_{\left\{M > 0\right\}}\cdot\sum^{M + 1}_{k = 2} X_{k}, M + 1\right)} = {\left(S_{f}, S_{d}, M + 1\right)} = S,
\end{aligned}
\end{displaymath}
\end{definition}

\noindent where $1_{\left\{A\right\}}$ is the indicator function for the set $\{A\}$.
The random elements $S_{f}$, $S_{d}$, and $M + 1$ respectively define the temporal distance between  flows (measured by the initial packets in each flow), the flow duration, and flow size.
If we instead supply a sequence of inter-renewals from a thinned (sub-sampled) flow, say $\widetilde{\bm{X}}$ which is a subset of the values of $\bm{X}$, as the argument for $\varphi$, this yields the \emph{sampled NetFlow} $\widetilde{S}={\left(\widetilde{S}_{f}, \widetilde{S}_{d}, \widetilde{M} + 1\right)}$. 
By construction, we have that $\widetilde{S}_{f}\geq S_{f}$, $\widetilde{S}_{d}\leq S_{d}$, and $\widetilde{M}\leq M+1$ when $\widetilde{\bm{X}}$ is sampled from the renewal process which generated $\bm{X}$.
The set of NetFlows for an entire session is obtained by applying the aggregation function $\varphi$ to each flow.

Having established an appropriate mathematical equivalent for a NetFlow, we can now derive the NetFlow likelihood.

\begin{proposition}[NetFlow likelihood] \label{prop:netflow_likelihood}
Let $S=(S_f,S_d,M+1)$ be the NetFlow symbol obtained from a sequence of inter-renewals $\bm{X}=\left(X_1,\ldots,X_{M+1}\right)$ of random length via Definition~\ref{def:flow_symbol}.
Suppose that the $i$-th inter-renewal $X_{i}$ has density $g_{i}{\left(x_{i};\theta_{i}\right)}$, for $i = 1,\ldots,M+1$. 
Then the realised NetFlow likelihood is 
\begin{equation} \label{eq:netflow_likelihood}
    \mathcal{L}_{S}{\left(s;\bm{\theta}\right)} \propto g_{1}{\left(s_{f};\theta_{1}\right)}\,\mathcal{G}{\left(s_{d};\bm{\theta}\right)}\,p_{M}{\left(m + 1;\nu\right)},
\end{equation}
\noindent where $p_{M}$ is the probability mass function for the random flow size, 
\begin{displaymath}
\mathcal{G}{\left(s_{d};\bm{\theta}\right)} = {\left(g_{2}*\cdots*g_{m + 1}\right)}{\left(s_{d};\bm{\theta}\right)},  
\end{displaymath}
$\bm{\theta} = (\theta_{1},\ldots,\theta_{m+1})$, and $f*g$ denotes the convolution of the densities $f$ and $g$. 
\end{proposition}

\begin{proof}
See Appendix~\ref{proof:netflow_likelihood}.
\end{proof}

\noindent The NetFlow likelihood \eqref{eq:netflow_likelihood} is a representative model for typically recorded NetFlow data.
If packet arrivals are defined by a simple renewal process, then $g_i=g$, for all $i$.
If, however, the packets arrive via a Bartlett--Lewis process, so that $X_{1}$ has Exponential density $f(\cdot;\lambda)$ and $X_{2},\ldots,X_{M + 1}$ have some common density $g(x;\theta)$, then the realised NetFlow likelihood is 
\begin{displaymath} 
    \mathcal{L}_{S}{\left(s;\bm{\theta}\right)} \propto f{\left(s_{f};\lambda\right)}\,g^{*(m)}{\left(s_{d};\theta\right)}\,p_{M}{\left(m + 1; \nu\right)},
\end{displaymath}
\noindent where $g^{*(k)}$ is the $k$-fold self-convolution of $g$, and $\bm{\theta}=(\lambda,\theta)$.

The NetFlow likelihood for an entire session is obtained by taking the product of NetFlow likelihoods for individual flows, since flows are independent in the Bartlett--Lewis process.
Hence, by letting $s_{i}$ be the observed NetFlow symbol for the $i$-th flow and writing $\bm{s} = (s_{1},\ldots,s_{n})$, we can express the \emph{sessional} NetFlow likelihood by
\begin{displaymath} 
    \mathcal{L}_{S}{\left(\bm{s};\bm{\theta}\right)} \propto \prod^{n}_{i = 1} f{\left(s_{f_{i}};\lambda\right)}\,g^{*\left(m_{i}\right)}{\left(s_{d_{i}};\theta\right)}\,p_{M}{\left(m_{i} + 1;\nu\right)}. 
\end{displaymath}

\noindent The NetFlow likelihood derived in Proposition~\ref{prop:netflow_likelihood} assumes that there is no underlying packet thinning.
Accordingly, the likelihood requires some modification when traffic is thinned prior to aggregation.
The following proposition considers the sampled NetFlow symbol obtained by sampling the original inter-renewal sequence.

\begin{proposition}[Sub-sampled NetFlow likelihood] \label{prop:sampled_netflow_likelihood}
Consider the renewal process defined in Proposition~\ref{prop:netflow_likelihood}.
Suppose that each arrival is retained with some constant known probability $q\in(0,1)$. 
The sub-sampled NetFlow likelihood is then
\begin{equation} \label{eq:sampled_netflow_likelihood}
    \mathcal{L}_{\widetilde{S}}\left(\widetilde{s};\bm{\theta}\right) \propto \sum_{m + 1 \geq \widetilde{m}} \sum^{m + 2 - \widetilde{m}}_{j = 1} \sum^{m + 1 - j}_{k = \widetilde{m} - 1} p_{M}{\left(m + 1;\nu\right)}\,\tau_{m}{(\widetilde{m};q)}\,\upsilon_{m,k}{\left(\widetilde{m}\right)}\,\mathcal{G}_{j,k}{\left(\widetilde{s};\bm{\theta}\right)},
\end{equation}
\noindent where $\tau_{m}{\left(\widetilde{m};q\right)} = \binom{m + 1}{\widetilde{m}}q^{\widetilde{m}}\left(1 - q\right)^{m + 1 - \widetilde{m}}$, $\upsilon_{m,k}{\left(\widetilde{m}\right)} = \binom{k - 1}{\widetilde{m} - 2} / \binom{m + 1}{\widetilde{m}}$, and 
\begin{displaymath}
  \mathcal{G}_{j,k}{\left(\widetilde{s};\bm{\theta}\right)} = {\left(g_{1}*\cdots*g_{j}\right)}{\left(\widetilde{s}_{f};\bm{\theta}\right)}\cdot{\left(g_{j + 1}*\cdots*g_{j + k}\right)}{\left(\widetilde{s}_{d};\bm{\theta}\right)}.
\end{displaymath}
\end{proposition}

\begin{proof}
See Appendix~\ref{proof:sampled_netflow_likelihood}.
\end{proof}

\noindent The probability mass function $\tau_{m}$ is the Binomial probability of sampling $\widetilde{m}$ packets from the original flow with $m + 1$ packets.
The random variables $J$ and $K$ (and indices $j$ and $k$) respectively define the location of the first sampled packet with respect to the complete sequence of inter-renewals, and the number of inter-renewals between the first and last sampled packets.
Hence, in computing $\mathcal{G}_{j,k}{\left(\widetilde{s};\bm{\theta}\right)}$, we observe a $j$-fold convolution of densities for the first sampled inter-renewal time, and a $k$-fold convolution of densities for the sampled flow duration.
The probability mass function $\upsilon_{m,k}$ conditionally defines the number of sampled flows whose first and last sampled packets are respectively the $j$-th and $(j + k)$-th packets.

The computational cost of evaluating \eqref{eq:sampled_netflow_likelihood} is primarily determined by the structure of the original flow sizes $m$ and the sub-sampling rate $q$.
If $p_{M}$ is known {\em a priori} and does not give significant mass to large flow sizes $m$, then \eqref{eq:sampled_netflow_likelihood} can be computed reasonably efficiently.
However, computational overheads will be high, and may require approximation, if large latent flow sizes are considered. 
If unknown, $p_{M}$ can be estimated using, for example, the Negative Binomial distribution or series inversion \citep{Hohn06}.

As before, we can obtain the Bartlett--Lewis representation of the sub-sampled NetFlow likelihood by setting $g_{1}$ to be the Exponential density and letting $g_{k} = g$, for $k \geq 2$, so that
\begin{displaymath} 
    \mathcal{G}_{j,k}{\left(\widetilde{s};\bm{\theta}\right)} = {\left(g_{1}*g^{*(j - 1)}\right)}{\left(\widetilde{s}_{f};\bm{\theta}\right)}\cdot g^{*(k)}{\left(\widetilde{s}_{d};\bm{\theta}\right)}.
\end{displaymath}

\noindent A sessional sampled NetFlow likelihood requires an average of $n!$ combinations of \eqref{eq:sampled_netflow_likelihood} since the ordering of flows cannot be determined from sampled arrivals without additional marks.

It is sufficient, however, to only consider the set of independent sampled flow durations if our attention is restricted to the packet-level model.
The total available information will, however, be underutilised since the set of $\widetilde{S}_{f}$ are not used.
The restricted sub-sampled NetFlow likelihood then simplifies to
\begin{displaymath}
    \mathcal{L}_{\widetilde{S}}{\left(\widetilde{s};\theta\right)} \propto \sum_{m + 1 \geq \widetilde{m}} \sum^{m}_{k = \widetilde{m} - 1} p_{M}{\left(m + 1;\nu\right)}\,\tau_{m}{\left(\widetilde{m};q\right)}\,\upsilon_{m,k}{\left(\widetilde{m}\right)}\,g^{*(k)}{\left(\widetilde{s}_{d};\theta\right)},
\end{displaymath}
\noindent with $\upsilon_{m,k}{\left(\widetilde{m}\right)} = \left(m + 1 - k\right)\binom{k - 1}{\widetilde{m} - 2} / \binom{m + 1}{\widetilde{m}}$.
By letting $\widetilde{s}_{i}$ be the $i$-th sub-sampled NetFlow and writing $\widetilde{\bm{s}} = {\left(\widetilde{s}_{1},\ldots,\widetilde{s}_{n}\right)}$, the \emph{restricted} sessional sampled NetFlow likelihood becomes
\begin{displaymath} 
    \mathcal{L}_{\widetilde{S}}{\left(\widetilde{\bm{s}};\theta\right)} \propto \prod^{n}_{i = 1} \mathcal{L}_{\widetilde{S}}{\left(\widetilde{s}_{i};\theta\right)}.
\end{displaymath}

\subsection{The NetFlow estimators}
The natural estimators are the parameters $\hat{\theta}_{S}$ and $\hat{\theta}_{\widetilde{S}}$ which respectively maximise the log-likelihoods

\begin{equation} \label{eq:netflow_log_likelihood}
\begin{aligned}
  \ell_{n}{\left(\bm{s};\theta\right)} := \frac{1}{n}\sum^{n}_{i = 1}\log{\mathcal{L}_{S}{\left(s_{i};\theta\right)}} 
  \quad \text{and} \quad 
  \widetilde{\ell}_{n}{\left(\widetilde{\bm{s}};\theta\right)} := \frac{1}{n}\sum^{n}_{i = 1}\log{\mathcal{L}_{\widetilde{S}}{\left(\widetilde{s}_{i};\theta\right)}}.
\end{aligned}
\end{equation}

\section{Attributes of the estimator} \label{sec:attributes_of_the_estimator}
It is important to determine if the NetFlow \textsc{mle} provides meaningful information about the \emph{true} parameter $\theta_{0}$ since there is an implicit loss of information when summarising complete session data to its NetFlow summaries.
More specifically, we wish to determine how efficiently the NetFlow \textsc{mle} converge to $\theta_{0}$, if at all.

We show that both NetFlow \textsc{mle}s are consistent, so that they converge to $\theta_{0}$ in probability.
We also provide bounds on the quantity of summarised data (i.e.~the number of NetFlows) required to obtain a desired level of efficiency relative to the standard \textsc{mle}.
Although specific to renewal processes, our development of consistency and relative information loss is novel and extends existing \textsc{sda} theory, since these properties are absent in the established \textsc{sda} literature, e.g.~\citep{Beranger18, Billard11, Lin17, Whitaker19a, Zhang19}.

We restrict our analysis to the packet level in order to avoid the factorial growth in computation required to also consider flow parameters.
The following results assume that the packet renewal model $g_X$ and its sequence of self-convolutions $g^{*(k)}_{X}$ satisfy standard regularity conditions \citep[p.~449]{Lehmann98}.
They also require that the series
\begin{equation} \label{eq:series_convergence}
    \xi(x;\theta) = \sum^{\infty}_{i = 1} \upsilon_{m,i}(\widetilde{m})\,g^{*(i)}_{X}{\left(x;\theta\right)}
\end{equation}
\noindent is jointly uniformly convergent in $\mathcal{X}\times\Theta$ when the possible flow sizes are unbounded.
Some relaxations of these assumption can be made, but at a cost of complicating the proofs.

\subsection{Consistency of the NetFlow estimator}
The following proposition extends the consistency of the standard \textsc{mle} for $\theta_{0}$ to the NetFlow \textsc{mle}.

\begin{proposition} \label{prop:netflow_estimator_consistency}
    The NetFlow \textsc{mle}s $\hat{\theta}_{S}$ and $\hat{\theta}_{\widetilde{S}}$ are consistent for $\theta_{0}\in\Theta^{0}$.
\end{proposition}

\begin{proof}
See Appendix~\ref{proof:netflow_estimator_consistency}.
\end{proof}

\subsection{Efficiency of the NetFlow estimator}

The standard \textsc{mle} is asymptotically Normally distributed under some weak regularity conditions with variance ${\left(NH\right)}^{-1}$, where $N$ is the number of observed packet inter-renewals in the session and $H$ is the Fisher information of the inter-renewal distribution evaluated at $\theta_{0}$ \citep{Hopfner14}. 
We can adapt this result for the NetFlow \textsc{mle}s by considering the Fisher information of the marginal densities
\begin{equation} \label{eq:flow_duration_marginal_density}
\begin{aligned}
  f_{S_{d}}{\left(s_{d};\theta\right)}                         &= {\sum_{m+1 \geq 2}} p_{M}{\left(m + 1;\nu\right)}\,g^{*(m)}{\left(s_{d};\theta\right)} \quad\text{and}\\
  f_{\widetilde{S}_{d}}{\left(\widetilde{s}_{d};\theta\right)} &= {\sum_{m + 1 \geq 2}\sum^{m + 1}_{\widetilde{m} = 2}\sum^{m}_{k = \widetilde{m} - 1}} p_{M}{\left(m+1;\nu\right)}\,\tau_{m,\widetilde{m}}{\left(q\right)}\,\upsilon_{m,k,\widetilde{m}}\,g^{*(k)}{\left(\widetilde{s}_{d};\theta\right)}.
\end{aligned}
\end{equation}
\noindent Slower convergence of the NetFlow \textsc{mle} is expected due to the implicit loss of information in NetFlow summarisation and under packet sampling.
It naturally follows that we should aim to identify the session size $n$ which will yield a given degree of fidelity of the NetFlow estimators compared to the standard \textsc{mle}.

Distributional metrics and divergences are natural paths to answering this question. 
However, a simpler approach is to consider the difference in the (asymptotic) variance of the estimators, since they are each asymptotically Normally distributed with common mean.

The covariance matrix $\Sigma$ of a $d$-dimensional Normal random variable induces a hyper-ellipsoid whose semi-axes lengths are proportional to the square root of its eigenvalues $\left\{\chi_{i}\right\}^{d}_{i = 1}$ \citep{Tong90}.
Hence, the efficiency of the standard \textsc{mle} and the NetFlow \textsc{mle} can be compared geometrically through the volumes of their induced hyper-ellipsoids
\begin{displaymath} 
    V_{d} \propto \sqrt{\prod^{d}_{i = 1} \chi_{i}} = \vert\Sigma\vert^{1 / 2}.
\end{displaymath}

\noindent If the standard \textsc{mle} is computed over $k$ flows -- so that $N = k\overline{M}$, where $\overline{M}$ is the mean number of packet-level inter-renewals per flow -- then the \textsc{mle} has covariance matrix $\Sigma = {\left(k\overline{M}H\right)}^{-1}$.
Similarly, if we let $I$ be the information matrix of the flow duration for either of the marginal densities \eqref{eq:flow_duration_marginal_density} and compute the NetFlow \textsc{mle}s over $n$ NetFlows, then we can express the covariance matrix of the NetFlow estimator as $\Upsilon = {\left(nI\right)}^{-1}$.

The following proposition defines a suitable range for the number of NetFlows which should be observed so that the standard and NetFlow \textsc{mle}s have comparable efficiency by restricting the relative volumes of the induced hyper-ellipsoids of $\Sigma$ and $\Upsilon$.

\begin{proposition} \label{prop:netflow_bounds}
    Respectively define $H$ and $I$ to be the information matrices for the density $g_{X}$ and one of the densities in \eqref{eq:flow_duration_marginal_density}, and let $\overline{M}$ be the mean number of inter-renewals per flow.
    The number $n$ of NetFlows that one should observe for the NetFlow \textsc{mle} to attain $\varepsilon$\% relative efficiency of the standard \textsc{mle} with probability $1 - \eta$ is
    \begin{equation} \label{eq:netflow_bounds}
        n\in\left(\pm k\sqrt[d]{\frac{\vert H\vert / \vert I\vert}{{\left(1\pm\varepsilon\right)}^{2}}}\log{\left(\frac{2}{\eta}\E{\left[\mathrm{e}^{\pm \overline{M}}\right]}\right)}\right).
    \end{equation}
\end{proposition}

\noindent We typically only care for the lower bound in \eqref{eq:netflow_bounds} since greater efficiency is generally desirable and additional NetFlows can be obtained freely. 

\begin{proof}
See Appendix~\ref{proof:netflow_bounds}.
\end{proof}

\section{Exact inference with the Natural Exponential Family} \label{sec:optimal_model}

The NetFlow likelihood requires a user specified model $g_X$ for the inter-renewal times. 
Because the sufficient statistics to fit this model may not be perfectly captured in NetFlow aggregation, a natural question is to ask whether there is any family of models for which NetFlow aggregation does not forfeit any inferential capacity.
Below we show that the \emph{Natural Exponential Family} (\textsc{nef}) satisfies this criterion.

\begin{definition}[Natural Exponential Family] \label{def:natural_exponential_family}
The Natural Exponential Family is a sub-class of the Exponential Family of distributions whose natural parameter and sufficient statistic are the identity. 
Respectively writing $h(x)$ and $A(\theta)$ for the base measure and log-partition function, its density can be written as
\begin{equation} \label{eq:natural_exponential_density}
    f_{X}{\left(x;\theta\right)} = h{(x)}\exp{\left(\theta x - A(\theta)\right)}.
\end{equation}
\end{definition}

\noindent \citet{MorrisLock06} show that the $k$-fold convolution of \textsc{nef} has density
\begin{equation} \label{eq:natural_exponential_family_convolution_density}
    f^{*(k)}_{X}{\left(x;\theta\right)} = h_{k}(x)\exp\left(\theta x - kA(\theta)\right).
\end{equation}

\noindent Evidently, convolutions of \textsc{nef} random variables scale the sufficient statistic identically to the NetFlow aggregation. As a result, we have the  following lemma.

\begin{lemma} \label{lem:identical_inference}
Inference for the standard and NetFlow \textsc{mle}s is identical when the supplied model $g_{X}$ is of the Natural Exponential Family.
\end{lemma}

\begin{proof}
See Appendix~\ref{proof:identical_inference}.
\end{proof}

\noindent Although Lemma~\ref{lem:identical_inference} only holds for networks \emph{without} packet thinning, i.e.~$q = 1$, the \textsc{nef} performs comparatively better in practice than other models with packet-thinned traffic. 

This result provides a practical restriction on the optimality of model fitting since, in most cases, the sufficient statistics for distributions outside of the \textsc{nef} are not derivable from the standard NetFlow summary.
For example, the sufficient statistic for the shape parameter $\alpha$ of the Gamma distribution is $\log x$.
The exact sufficient statistics of convolved Gamma random variables cannot be recovered from the flow duration $S_{d} = \sum^{M+1}_{i = 2} X_{i}$.

\section{Simulations} \label{sec:simulations}
We now explore the performance of the NetFlow \textsc{mle}s on various synthetic networks.
\citet{Hohn03} show through empirical validation, that packet transit can be described by the Poisson--Gamma class of the Bartlett--Lewis process, so that the temporal distance between consecutive flows and packets are respectively Exponentially and Gamma distributed. 
\citet{Gonzales10} provide an algorithm to generate this process.
These models were utilised by \citet{Antunes18}, and so we adopt them here for our simulated analyses.
Characterising flow-level traffic as Poisson processes admits many attractive features, including independence of flows, memorylessness for successive flows, and uniformity of flow arrivals.
In addition, establishing packet-level arrivals by finite Gamma renewal processes permits grouped burstiness and packet delay.
Evaluation of the session NetFlow log-likelihood \eqref{eq:netflow_log_likelihood} requires computing convolutions of the supplied model $g_{X}$.
While computing convolutions of random variables can be expensive, the Gamma distribution is closed under convolutions.
The parameters used to generate our simulated datasets are adopted from the empirical analyses of \citet{Antunes18}.
We restrict our attention to assessing packet level characteristics in order to avoid large computational overheads.

\begin{example}[Method-of-moments comparison] \label{ex:moments}
In this example, we compare the performance of the NetFlow \textsc{mle} against the method-of-moments estimators of \citet{Hohn03}.
Our results show that the moment-based estimators fail to converge for the established network, unlike the NetFlow \textsc{mle}, and we provide a na\"{i}ve modification for them to converge.
We discuss conditions under which the estimators of \citet{Hohn03} converge, and contribute comparative analyses for these conditions in Appendix~\ref{app:simulations}.

We first describe the data generating process.
The Zeta distribution is often used to model the flow sizes since network traffic is typically dominated by many small \emph{mouse} flows and few large \emph{elephant} flows \citep{Hohn03, Hohn06, Antunes18}.
The empirical analysis of \citet{Antunes18} yields a flow size shape parameter of $k = 1.02$.
However, it utilises a distributional approximation by defining the true mean flow size to be $k(k - 1)^{-1}$.
This quantity is in fact the mean for the Pareto distribution with minimum flow size $M + 1 = 1$, a continuous analogue for the Zeta distribution whose mean is $\zeta{(\kappa - 1)} / \zeta{(\kappa)}$, where $\zeta{(\cdot)}$ is the Riemann zeta function.
Through reverse-engineering to compute the average flow size $\overline{m + 1} = 51$ and optimisation of the Zeta mean, we obtain a shape parameter $\kappa = 2.012085$, implying that the flow sizes have finite mean but infinite variances since $\kappa\in{(2, 3)}$.
We assign a finite Gamma renewal process for the packet-level process so that the distance between consecutive packets (within a flow) are Gamma distributed with parameters ${\left(\alpha_{0},\beta_{0}\right)} = {\left(0.6, 526.32\right)}$.

To generate a flow $\bm{X} = (X_{1},\ldots X_{M})$, we first sample the random flow size $M + 1\sim\text{Zeta}$, and then generate a sequence of $M$ independent $\text{Gamma}\left(\alpha_{0},\beta_{0}\right)$-distributed inter-renewals.
The NetFlow $S = \left(S_{d}, M\right)$ with $S_{d} = \sum^{M}_{i = 1} X_{i}$ is then computed.
The tuples $\bm{X}$ and $S$ respectively define the complete and summarised data for a single flow.
Realised data for an entire session is represented by $\bm{x}_{1}, \ldots,\bm{x}_{n}$ and $s_{1},\ldots,s_{n}$.
Though we do not assess the flow-level model, we otherwise note that its standard and NetFlow \textsc{mle} would coincide.

\citet{Hohn03} respectively define point estimators for $\alpha_{0}$ and $\beta_{0}$ through the empirical coefficient of variation and a weighted average packet intensity. 
The shape estimator is explicitly defined by $\hat{\alpha} = \left(\overline{x} / \hat{\sigma}_{x}\right)^{2}$, where $\overline{x}$ and $\hat{\sigma}_{x}$ are respectively the empirical mean and standard deviation of all packet-level inter-renewals.
The rate estimator is then defined as $\hat{\beta} = \hat{\alpha}\sum^{n}_{i=1}w_{i}\varrho_{i}$, with  weights $w_{i} = m_{i} / \sum^{n}_{j = 1} m_{j}$, and packet intensities $\varrho_{i} = m_{i} / s_{d_i}$.
While computationally fast, this approach is sequential and requires the complete set of inter-renewals. 
These estimators are also specific to the Gamma distribution.

We provide a \textsc{sda} equivalent of the moments-based estimators of \citet{Hohn03} by defining moments-based estimators which use only NetFlow data.
Let $Y_{i} = S_{d_i} / M_{i}$ be the mean inter-renewal of the $i$-th flow.
We define the estimator $\check{\alpha} = {\left(\overline{y} / \hat{\sigma}_{y}\right)}^{2}\,\overline{1 / m}$, where $\overline{y}$ and $\hat{\sigma}_{y}$ are respectively the sample mean and standard deviation of the mean inter-renewal times, and $\overline{1 / m} = n^{-1}\sum_{i=1}^{n} 1/m_{i}$ is the mean of the reciprocal flow sizes (minus 1).
Substituting $\check{\alpha}$ for $\hat{\alpha}$ into the previous rate estimator yields $\check{\beta} = \check{\alpha}\sum^{n}_{i=1}w_{i}\varrho_{i}$.
We also compute na\"{i}ve estimators $\hat{\beta}^{*} = \hat{\alpha} / \overline{x}$ and $\check{\beta}^{*} = \check{\alpha} / \overline{z}$, where $\overline{z} = \sum^{n}_{i=1} y_{i} / \sum^{n}_{i = 1} m_{i}$.
Note that $\overline{x} = \overline{z}$.

Table~\ref{tab:moments_scenario_1} presents the average point estimates and standard errors for each of the estimators from $T = 10^{3}$ replicates of the network for sessions of various size $n = 10^{2k}$, where $k = 0,1,2,3$.
Note that $\check{\alpha}$ is not computable for sessions with only $n = 1$ flow since the variance of a single value is zero.
\begin{table}
    \caption{\small{
    Example~\ref{ex:moments}: Mean point estimates (and standard errors) of $(\alpha,\beta)$ from $T = 10^{3}$ replicate synthetic sessions of size $n$. 
    True values are $(\alpha_{0},\beta_{0})=(0.6, 526.32)$.
    }}
    \label{tab:moments_scenario_1}
    \centering
    \begin{tabular}{@{}rrrrr@{}} \toprule
                        & \multicolumn{4}{c}{$n$}                                                                                                   \\ \cmidrule{2-5}
                        & $10^{0}$                     & $10^{2}$                     & $10^{4}$                     & $10^{6}$                     \\ \midrule
    \multicolumn{5}{l}{\small{\emph{Method-of-moments}}}                                                                                            \\
    $\hat{\alpha}$      & 3.68\p                       & 0.61\p                       & 0.60\p                       & 0.60\p                       \\
                        & (1.09)                       & $\left({\sim}10^{-3}\right)$ & $\left({\sim}10^{-4}\right)$ & $\left({\sim}10^{-5}\right)$ \\ \addlinespace[1ex]
    $\hat{\beta}$       & ${\sim}10^{3}$\p             & ${\sim}10^{4}$               & ${\sim}10^{5}$\p             & ${\sim}10^{7}$\p             \\
                        & $\left({\sim}10^{3}\right)$  & $\left({\sim}10^{4}\right)$  & $\left({\sim}10^{4}\right)$  & $\left({\sim}10^{6}\right)$  \\ \addlinespace[1ex]
    $\hat{\beta}^{*}$         & ${\sim}10^{3}$\p             & 540.17\p                     & 526.38\p                     & 526.33\p               \\
                        & $\left({\sim}10^{3}\right)$  & (2.68)                       & (0.18)                       & (0.02)                       \\ \addlinespace[1ex]
    \multicolumn{5}{l}{\small{\emph{NetFlow method-of-moments}}}                                                                                    \\
    $\check{\alpha}$    & ---                          & 0.73\p                       & 0.60\p                       & 0.60\p                       \\ 
                        & ---                          & $\left({\sim}10^{-3}\right)$ & $\left({\sim}10^{-4}\right)$ & $\left({\sim}10^{-5}\right)$ \\ \addlinespace[1ex] 
    $\check{\beta}$     & ---                          & ${\sim}10^{4}$\p             & ${\sim}10^{5}$\p             & ${\sim}10^{7}$\p             \\
                        & ---                          & $\left({\sim}10^{3}\right)$  & $\left({\sim}10^{4}\right)$  & $\left({\sim}10^{6}\right)$  \\ \addlinespace[1ex]
    $\check{\beta}^{*}$ & ---                          & 645.13\p                     & 527.78\p                     & 526.28                       \\
                        & ---                          & (3.15)                       & (0.33)                       & (0.03)\p                     \\ \addlinespace[1ex]
    \multicolumn{5}{l}{\small{\emph{NetFlow \textsc{mle}}}}                                                                                         \\
    $\hat{\alpha}_{S}$  & ${\sim}10^{30}$\p            & 0.65\p                       & 0.60\p                       & 0.60\p                       \\
                        & $\left({\sim}10^{30}\right)$ & $\left({\sim}10^{-3}\right)$ & $\left({\sim}10^{-4}\right)$ & $\left({\sim}10^{-5}\right)$ \\ \addlinespace[1ex]
    $\hat{\beta}_{S}$   & ${\sim}10^{35}$\p            & 569.93\p                     & 527.29\p                     & 526.30\p                     \\
                        & $\left({\sim}10^{35}\right)$ & (4.48)                       & (0.34)                       & (0.04)                       \\ \bottomrule 
    \end{tabular}
\end{table}
\begin{table}
    \caption{\small{
    Example~\ref{ex:moments}: Mean session information volume (megabytes) and computation time (milliseconds) over various session sizes $n$.
    }}
    \label{tab:metadata_moments_scenario_1}
    \centering
    \begin{tabular}{@{}lrrrr@{}} \toprule
                              & \multicolumn{4}{c}{$n$}                                             \\ \cmidrule{2-5}                          
                              & $10^{0}$ & $10^{2}$ & $10^{4}$ & $10^{6}$                           \\ \midrule
    \multicolumn{5}{l}{\small{\emph{Method-of-moments}}}                                                      \\
    Information (\textsc{mb}) & ${\sim}10^{-4}$ & ${\sim}10^{-3}$  & 1.83            & 100.78       \\
    Time $(ms)$               & ${\sim}10^{-1}$ & ${\sim}10^{-1}$  & 3\ppp           & 230\ppp      \\ \addlinespace[1ex]
    \multicolumn{5}{l}{\small{\emph{NetFlow method-of-moments}}}                                       \\
    Information (\textsc{mb}) & ---             & ${\sim}10^{-3}$  & 0.09            & 9.31         \\
    Time $(ms)$               & ---             & ${\sim}10^{-2}$  & ${\sim}10^{-1}$ & 12\ppp       \\ \addlinespace[1ex]
    \multicolumn{5}{l}{\small{\emph{NetFlow \textsc{mle}}}}                                            \\
    Information (\textsc{mb}) & ${\sim}10^{-4}$ & ${\sim}10^{-3}$  & 0.09            & 9.31         \\
    Time $(ms)$               & 1               & 1                & 53\ppp          & $5\,188$\ppp \\ \bottomrule
    \end{tabular}
\end{table}

The estimates for the shape parameter $\alpha$  converge well for all methods.
However, the moments-based rate estimators $\hat{\beta}$ and $\check{\beta}$ notably fail to converge.
This occurs as a result of estimating the rate $\beta$ using packet intensity $\varrho$, which has Inverse-Gamma distribution with infinite mean when its shape parameter is $\alpha' < 1$.
This arises when the flow size $M + 1 \leq 2$, which has approximate probability 0.94 here. 
We can ensure that $\alpha' > 1$ by either setting the inter-renewal shape parameter $\alpha_{0} > 1$ or by only recording flow sizes for which $M\alpha_{0} > 1$. 
These conditions are presented in detail in Appendix~\ref{app:simulations}.

As expected, the moments-based estimators $\hat{\alpha}$ and $\hat{\beta}^*$ converge faster than the NetFlow \textsc{mle} $\hat{\alpha}_{S}$ and $\hat{\beta}_{S}$ since they are evaluated using considerably more data. 
However, when comparing estimators which only use NetFlow data, we see that the NetFlow \textsc{mle} is comparable to its moments-based counterpart $\check{\alpha}$ and $\check{\beta}^*$.

Table~\ref{tab:metadata_moments_scenario_1} displays the average volume of information and evaluation time used to compute each point estimate, highlighting the trade-off between accuracy and computational speed.
For the moment-based estimators, those based on NetFlow data only $\left(\check{\alpha}, \check{\beta}^*\right)$ are typically at least one order of magnitude more efficient to compute than those based on the full flow data $\left(\hat{\alpha}, \hat{\beta}^*\right)$. 
This is the primary benefit of working with NetFlow session data.
In contrast, the NetFlow \textsc{mle} are more expensive to evaluate since they require an optimisation as opposed to simple arithmetic computation of a statistic.
However, the NetFlow \textsc{mle} can be computed arithmetically for some choices of model, e.g.~the Exponential distribution.
\end{example}

\begin{example}[Gamma packet model with thinning] \label{ex:gamma}
In the previous example, we generated a network with complete packet retention.
However, complete packet retention is rare in modern networks.
The simple moments-based estimators $(\hat{\alpha}, \hat{\beta}, \hat{\beta}^*)$ and their NetFlow equivalents are not valid outside of this setting as they are unable to account for packet thinning.
However, the sampled NetFlow \textsc{mle} can be practically applied to both  thinned and summarised network data.
We explore its performance here under various sampling rates. 

As previously discussed, the computational cost of the sub-sampled NetFlow \textsc{mle} is principally determined by the cardinality of the flow sizes.
Accordingly, we define a small sample space $\mathcal{M} = \left\{11,\,101,\,1\,001\right\}$.
The feasible flow sizes are chosen to approximate mouse $(M=10)$ and elephant $(M = 10^{3})$ flows.
Flow sizes are randomly sampled from a Zipf distribution with shape $\omega = 1$ to mimic the dominance of mouse flows such that
\begin{displaymath}
p_{M}{\left(11\right)} = \frac{6}{11},\;p_{M}{\left(101\right)} = \frac{3}{11},\text{ and } p_{M}{\left(1\,001\right)} = \frac{2}{11}.
\end{displaymath}

\noindent We again generate flows whose packet-level arrivals follow a finite Gamma renewal process with parameters $(\alpha_{0}, \beta_{0}) = (0.6, 526.32)$.
Each packet arrival is independently recorded with probability $q$ (or discarded with probability $1-q$) to form the sequence of sampled arrival times.
The sampled NetFlows $\widetilde{s}_i$ are then computed from these sampled arrival times using Definition~\ref{def:flow_symbol}.
In practical applications, it is feasible to wait a trivial amount of time until we obtain any desired number of sampled NetFlows since network traffic is so voluminous.

In addition to a range of sampling rates $q$, we also compute the NetFlow \textsc{mle} for sessions of size $n = 10^{k}$, for $k = 0,\ldots,3$ and $n_{\text{min}}$, the lower bound in \eqref{eq:netflow_bounds} with $\varepsilon = \eta = 0.1$.
In plain language, $n_{\text{min}}$ is the minimum number of flows needed for the NetFlow \textsc{mle} to have efficiency within 10\% of the standard \textsc{mle} 90\% of the time.
Computing $n_{\text{min}}$ analytically is quite difficult since the information matrix $I$ must be determined from the marginal mixture densities in \eqref{eq:flow_duration_marginal_density}. 
However, through numerical differentiation and Monte-Carlo integration, we respectively estimate that, approximately, $n_{\text{min}} = 81,\,106,\,551,\,5\,065,$ and $6\,507$ for the sampling rates $q = 10^{-k}$, for $k = 0,\ldots,4$.

As a point of reference, we also compute the standard \textsc{mle} for a single complete flow, i.e.~with sampling rate $q = 1$.
Table~\ref{tab:gamma} presents the average (sampled) NetFlow \textsc{mle} and standard errors for $T = 10^{3}$ replicate analyses under various session sizes $n$ and sampling rates $q$.
The rightmost column indicates the average number of seconds needed to compute the NetFlow \textsc{mle} from $n_{\text{min}}$ NetFlows.
\begin{table}
    \caption{\small{
    Example~\ref{ex:gamma}: mean (sampled) NetFlow \textsc{mle} (and standard errors) of $(\alpha,\beta)$ for a synthetic network with a packet-level Gamma process. 
    Estimates are obtained over a range of session sizes $n$ and packet thinning rates $q$. 
    Presented values are obtained from $T = 10^{3}$ replicate datasets with non-trivial flows, i.e.~$M,\widetilde{M} \geq 2$. 
    True values are $(\alpha_0,\beta_0)=(0.6, 526.32)$. 
    The right-most column shows the average time (seconds) to compute the (sampled) NetFlow \textsc{mle} using $n_{\text{min}} = 81,\,106,\,551,\,5\,065,$ and $6\,507$ NetFlows for respective sampling rates $q = 10^{-k}$, where $k = 0,\ldots,4$.
    }}
    \label{tab:gamma}
    \centering
    \begin{tabular}{@{}rrrrrrr@{}} \toprule
                         & \multicolumn{5}{c}{$n$}                                                                                                                                              \\ \cmidrule{2-6}
                         & $10^{0}$                     & $10^{1}$                    & $10^{2}$                     & $10^{3}$                     & $n_{\text{min}}$             & Time $(s)$ \\ \midrule
    \multicolumn{7}{l}{\small{\textsc{mle} $(q = 1)$}}                                                                                                                                          \\
    $\hat{\alpha}$       & 0.72\p                       & ---                         & ---                          & ---                          & ---                          & ---        \\
                         & (0.01)                       & ---                         & ---                          & ---                          & ---                          & ---        \\ \addlinespace[1ex]
    $\hat{\beta}$        & 703.74\p                     & ---                         & ---                          & ---                          & ---                          & ---        \\ 
                         & (14.69)                      & ---                         & ---                          & ---                          & ---                          & ---        \\ \addlinespace[1ex]
    \multicolumn{7}{l}{\small{NetFlow \textsc{mle}}}                                                                                                                                            \\
    \multicolumn{7}{l}{\small{$q = 1$}}                                                                                                                                                         \\ 
    $\hat{\alpha}_{S}$   & ${\sim}10^{29}$\p            & 0.84\p                      & 0.62\p                       & 0.60\p                       & 0.63                         & $10^{-3}$  \\
                         & $\left({\sim}10^{28}\right)$ & (0.02)                      & $\left({\sim}10^{-3}\right)$ & $\left({\sim}10^{-4}\right)$ & $\left({\sim}10^{-3}\right)$ &            \\ \addlinespace[1ex]
    $\hat{\beta}_{S}$    & ${\sim}10^{32}$\p            & 737.37\p                    & 544.33\p                     & 528.88\p                     & 551.38\p                     &            \\ 
                         & $\left({\sim}10^{31}\right)$ & (14.72)                     & (2.45)                       & (0.76)                       & (3.07)                       &            \\ \addlinespace[1ex]
    \multicolumn{7}{l}{\small{$q = 10^{-1}$}}                                                                                                                                                   \\ 
    $\tilde{\alpha}_{S}$ & ${\sim}10^{29}$\p            & 29.22\p                     & 0.64\p                       & 0.60\p                       & 0.63\p                       & 5          \\
                         & $\left({\sim}10^{28}\right)$ & (12.13)                     & $\left({\sim}10^{-3}\right)$ & $\left({\sim}10^{-3}\right)$ & $\left({\sim}10^{-3}\right)$ &            \\ \addlinespace[1ex] 
    $\tilde{\beta}_{S}$  & ${\sim}10^{33}$\p            & ${\sim}10^{4}$\p            & 561.20\p                     & 529.01\p                     & 553.17\p                     &            \\
                         & $\left({\sim}10^{32}\right)$ & $\left({\sim}10^{4}\right)$ & (3.89)                       & (1.09)                       & (3.89)                       &            \\ \addlinespace[1ex]
    \multicolumn{7}{l}{\small{$q = 10^{-2}$}}                                                                                                                                                   \\ 
    $\tilde{\alpha}_{S}$ & ${\sim}10^{28}$\p            & ${\sim}10^{3}$\p            & 5.49\p                       & 0.63\p                       & 0.66\p                       & 20         \\ 
                         & $\left({\sim}10^{27}\right)$ & (175.05)                    & (0.88)                       & $\left({\sim}10^{-3}\right)$ & $\left({\sim}10^{-3}\right)$ &            \\ \addlinespace[1ex] 
    $\tilde{\beta}_{S}$  & ${\sim}10^{32}$\p            & ${\sim}10^{6}$\p            & ${\sim}10^{3}$\p             & 552.61\p                     & 575.42\p                     &            \\
                         & $\left({\sim}10^{32}\right)$ & $\left({\sim}10^{5}\right)$ & (765.00)                     & (3.68)                       & (6.24)                       &            \\ \addlinespace[1ex]
    \multicolumn{7}{l}{\small{$q = 10^{-3}$}}                                                                                                                                                   \\
    $\tilde{\alpha}_{S}$ & ${\sim}10^{28}$\p            & ${\sim}10^{3}$\p            & 188.22\p                     & 2.59\p                       & 0.66\p                       & 97         \\
                         & $\left({\sim}10^{26}\right)$ & (143.03)                    & (10.73)                      & (0.23)                       & $\left({\sim}10^{-3}\right)$ &            \\ \addlinespace[1ex]
    $\tilde{\beta}_{S}$  & ${\sim}10^{31}$\p            & ${\sim}10^{6}$\p            & ${\sim}10^{5}$\p             & ${\sim}10^{3}$\p             & 578.06\p                     &            \\
                         & $\left({\sim}10^{30}\right)$ & $\left({\sim}10^{5}\right)$ & $\left({\sim}10^{3}\right)$  & (195.88)                     & (5.97)                       &            \\ \addlinespace[1ex]
    \multicolumn{7}{l}{\small{$q = 10^{-4}$}}                                                                                                                                                   \\
    $\tilde{\alpha}_{S}$ & ${\sim}10^{28}$\p            & ${\sim}10^{3}$\p            & 160.17\p                     & 2.77\p                       & 0.65\p                       & 115        \\
                         & $\left({\sim}10^{26}\right)$ & (120.78)                    & (8.97)                       & (0.26)                       & $\left({\sim}10^{-3}\right)$ &            \\ \addlinespace[1ex]
    $\tilde{\beta}_{S}$  & ${\sim}10^{31}$\p            & ${\sim}10^{6}$\p            & ${\sim}10^{5}$\p             & ${\sim}10^{3}$\p             & 573.49\p                     &            \\
                         & $\left({\sim}10^{30}\right)$ & $\left({\sim}10^{5}\right)$ & $\left({\sim}10^{3}\right)$  & (224.11)                     & (5.33)                       &            \\ \bottomrule
    \end{tabular}
\end{table}

The results show that the NetFlow \textsc{mle} converges to the true parameter values as the number of observed NetFlows increases, regardless of the level of thinning $q$.  
The results intuitively show that more NetFlows are required to achieve desired efficiency as the degree of packet thinning increases, i.e., as $q$ decreases.
This is also apparent in the increasing sequence of $n_{\text{min}}$ (with decreasing $q$) which aims to provide a constant degree of efficiency for each sampling rate.

The computational overheads of the NetFlow \textsc{mle} are inversely proportional to the sampling rates since lower sampling rates will yield smaller sampled flow sizes, which therefore necessitates integration over a larger number of latent (pre-sampled) flow sizes. 
\end{example}

\section{Real data analysis} \label{sec:real_data}
We now explore the performance of the NetFlow \textsc{mle} on real network data.
We obtain a \texttt{pcap} file\footnote{\texttt{equinix-nyc.dirA.20190117-1315500}} from \citep{Caida19} and assume that the trace captured all traffic through the network.
The \texttt{pcap} file is processed using Wireshark and extracted as a \texttt{csv} file for further analysis in \texttt{R}.
The trace captures 36\,197\,062 packets over the course of one minute. 
We identify 779\,788 non-trivial flows, i.e.~those containing at least 2 packets, from the full set of 1\,811\,255 flows.

We again restrict our analysis to the packet-level.
Timestamps are recorded at nanosecond granularity.
In some cases, inter-renewals will be recorded as 0 since consecutive packets may arrive within the threshold.
We assign these zero-valued inter-renewal times to be $10^{-7}$ seconds, on the scale of the smallest positive inter-renewal time $\delta\approx2.38\times10^{-7}$ seconds.

\subsection{Full packet retention $q = 1$}
We first provide an analysis assuming complete data.
The standard \textsc{mle} when fitting the Gamma model requires approximately 2 hours to compute and its fitted survivor function is presented in Figure~\ref{fig:data_survival_functions}(a).
It is clear from Figure~\ref{fig:data_survival_functions} that the Gamma model is inappropriate and hence, we consider alternative models.
Computing the sample mean and standard deviation of the log-transformed inter-renewals gives
\begin{displaymath} \label{eq:mle_point_estimate}
\begin{aligned}
    (\overline{\log(x)},\,\hat{\sigma}_{\log(x)}) &= \left(-8.0987,\, 4.5046\right),
\end{aligned}
\end{displaymath}

\noindent where $x$ denotes the packet inter-renewals.
These quantities correspond to the \textsc{mle} for the Log-Normal distribution.
Figure~\ref{fig:data_survival_functions}(a) shows that the Log-Normal model provides an appropriate fit for the packet inter-renewals.

The NetFlow \textsc{mle} for the Log-Normal model is not immediately accessible since there are no simple, closed-form convolutions of Log-Normal random variables.
We estimate the convolution through the \emph{Fenton--Wilkinson} approximation, which states that the $k$-fold convolution of a Log-Normal random variable with parameters $\mu$ and $\sigma$ can be approximated in the tail by a single Log-Normal random variable with parameters \citep{Marlow67}
\begin{displaymath}
\begin{aligned}
    \mu_{*} = \log(k) + \mu - \frac{\sigma^{2} - \sigma^{2}_{*}}{2}
    \quad \text{and } \quad 
    \sigma^{2}_{*} = \log{\left(\frac{\exp{\left(\sigma^{2}\right)} - 1}{k} + 1\right)}. 
\end{aligned}
\end{displaymath}

\noindent Unfortunately, this approximation cannot be readily substituted into the NetFlow likelihood since the session contains several mouse flows whose durations are too small to satisfy the tail approximation.
We remedy this by further aggregating the set of NetFlows into a single \emph{session NetFlow}, since the sum of independent  flow durations is sufficiently large.
To obtain the session NetFlow, we take the element-wise sum of all NetFlows.
Specifically, if we have a set of NetFlows $s_{i} = (s_{d_i},m_{i} + 1)$, for $i = 1,\ldots,n$, then the session NetFlow is 
$s_{*} = \left(\sum^{n}_{i = 1} s_{d_i}, n + \sum^{n}_{i = 1}m_{i}\right)$.
This \emph{two-step} approximation yields the NetFlow \textsc{mle} $\left(\hat{\mu}_{S}, \hat{\sigma}_{S}\right) = (-7.9511, 3.6684)$.

Figure~\ref{fig:data_survival_functions}(a) shows that the NetFlow \textsc{mle} slightly underestimates large scale inter-renewals, but is an otherwise satisfactory representation of the observed data.
Table~\ref{tab:real_data} (top two rows) presents the size of each dataset and times for computing the standard and NetFlow \textsc{mle}.
In this instance, computation for the \textsc{mle} is trivially faster since the point estimates can be expressed through simple arithmetic. 
However, we note that the byteage of the complete set of inter-renewals required to compute the standard \textsc{mle} is 26.1 times larger than the set of NetFlows, and 
that this larger dataset is typically not recorded in practice.
Accordingly, the NetFlow \textsc{mle} may be preferable when comparing the costs of capturing and processing such large volumes of data.

\subsection{Packet sampling with $q = 0.001$}
We consider a more realistic setting by analysing packet-thinned samples of the full dataset, with sampling probability $q = 0.001$.
To generate the sampled data, sampled flow sizes $\widetilde{m}$ are generated from a Binomial distribution for each observed flow size $m + 1$.
We then independently sample and sort $\binom{m + 1}{\widetilde{m}}$ random arrival times from the complete flow.
Trivial sampled NetFlows are discarded when $\widetilde{m}\leq1$.
The sampled NetFlows are then computed using Definition~\ref{def:flow_symbol}.
This procedure is statistically identical to, but more computationally efficient than, direct Bernoulli sampling and respectively produced 36\,189 sampled arrivals and 4\,553 sampled NetFlows.

A na\"{i}ve ``\textsc{mle}'' for the sampled inter-renewals was obtained by computing the mean and standard deviation of the log-transformed of sampled inter-renewals.
In order to compute the sampled NetFlow \textsc{mle}, we need to supply the \textsc{pmf} for the original flow size $p_{M}$.
Several methods exist to approximate $p_{M}$ from a set of sampled flow sizes, the simplest application of which is use of the Negative-Binomial distribution.
However, since our primary aim is to estimate the packet-process parameters, we compute an empirically derived \textsc{pmf}.
As previously noted, the computational overhead of the sampled NetFlow \textsc{mle} is dictated by the cardinality of the feasible flow sizes. 
To limit this, we construct an approximate flow size distribution by assuming that the original flow sizes are restricted to $M = \lceil j\times10^{k} \rceil$, where $j = 1, 2.5, 5$, and $k = 0,\ldots,5$, rounding each observed flow size to the nearest restricted flow size, and then deriving the approximate \textsc{pmf} from the proportion of flows rounded to each restricted flow size value.
The same procedure as in Section 7.1 is then used to obtain the sampled NetFlow \textsc{mle}. 

The results are presented in the middle rows of Table~\ref{tab:real_data} and in Figure~\ref{fig:data_survival_functions}(b).
The quantity $p$ indicates the proportion of the total (sub-sampled) dataset used to compute the point estimate.
For example, for the standard \textsc{mle}, $p = 1$ indicates that all the observed inter-renewals were used; whereas for the sampled NetFlow \textsc{mle}, $p = 10^{-3}$ indicates that only 0.1\% of the available sampled NetFlows were used.
The quantity $n$ shows the resulting number of datapoints used in the optimisation, which corresponds to $p$.
The presented values are the mean point estimate and standard error of $T = 10$ random samples from the full dataset.

From Figure~\ref{fig:data_survival_functions}(b) it is clear that na\"{i}vely fitting the Log-Normal distribution to the sampled inter-renewals (solid-grey line) fails to adequately describe the packet transit distribution.
The sampled NetFlow \textsc{mle} performs well, even with the approximations involved in computing the convolution of Log-Normal densities and the approximation of the flow size \textsc{pmf}.
In this case, $n=456$ provides a balance between estimator convergence and computational overheads.
The model fit is naturally poorer than when using the full dataset, although the full dataset is typically unavailable.
More accurate fits can be obtained by using finer approximations of $p_{M}$ and increased numbers of sampled NetFlows $n$. 
In summary, it is clear that the sampled NetFlow \textsc{mle} provides a practical and viable solution for modelling  network traffic, for which only aggregated NetFlow data from thinned traffic data are recorded.
\begin{table}
    \centering
    \caption{\small{
    Mean point estimates (and standard errors) for the real data application in Section~\ref{sec:real_data}.
    The quantity $p$ determines the proportion of total available information used to compute the point estimate.
    The number $n$ denotes the number of datapoints which correspond to $p$ and respectively refers to the number of inter-renewals and NetFlows over which the standard and NetFlow \textsc{mle} are computed.
    The right-most column presents the average time (in seconds) required to compute the point estimate.
    Estimates which correspond to $p < 1$ are averaged over $T = 10$ samples of the data set.
    }}
    \label{tab:real_data}
    \begin{tabular}{@{}lrrrrr@{}} \toprule
                          & Data volume     &            & Parameters &               & Times $(s)$ \\ \cmidrule{2-3} \cmidrule{4-5} 
                          & $p$             &  $n$       & $\mu$      & $\sigma$      &             \\ \midrule
    \multicolumn{6}{l}{$q = 1$}                                                                     \\
    NetFlow \textsc{mle}  & 1\phantom{.000} & 779 788    & -8.05\p    & 3.70\p       & 6            \\ 
                          &                 &            & ---        & ---          &              \\ \addlinespace[1ex]
    Standard \textsc{mle} & 1\phantom{.000} & 34 385 807 & -8.10\p    & 4.50\p       & 3            \\
                          &                 &            & ---        & ---          &              \\ 
    \multicolumn{6}{l}{$q = 0.001$}                                                                 \\
    NetFlow \textsc{mle}  & 0.001           & 5          & -7.41\p    & 2.24\p       & 410          \\
                          &                 &            & (0.64)     & (0.41)       &              \\ \addlinespace[1ex]
                          & 0.01\phantom{1} & 46         & -7.84\p    & 2.73\p       & 714          \\
                          &                 &            & (0.46)     & (0.29)       &              \\ \addlinespace[1ex]
                          & 0.1\phantom{01} & 456        & -8.90      & 3.20         & 2 718        \\
                          &                 &            & (0.03)     & (0.07)       &              \\
                          & 1\phantom{.000} & 4 553      & -8.79      & 2.92         & 13 938       \\ 
                          &                 &            & ---        & ---          &              \\ \addlinespace[1ex]
    Standard \textsc{mle} & 1\phantom{.000} & 19 511     & 0.22       & 2.42         & 0            \\ 
                          &                 &            & ---        & ---          &              \\ \bottomrule
    \end{tabular}
\end{table}

\begin{figure}
    \centering
    \subfigure[$q = 1$]{\includegraphics[scale = 0.8]{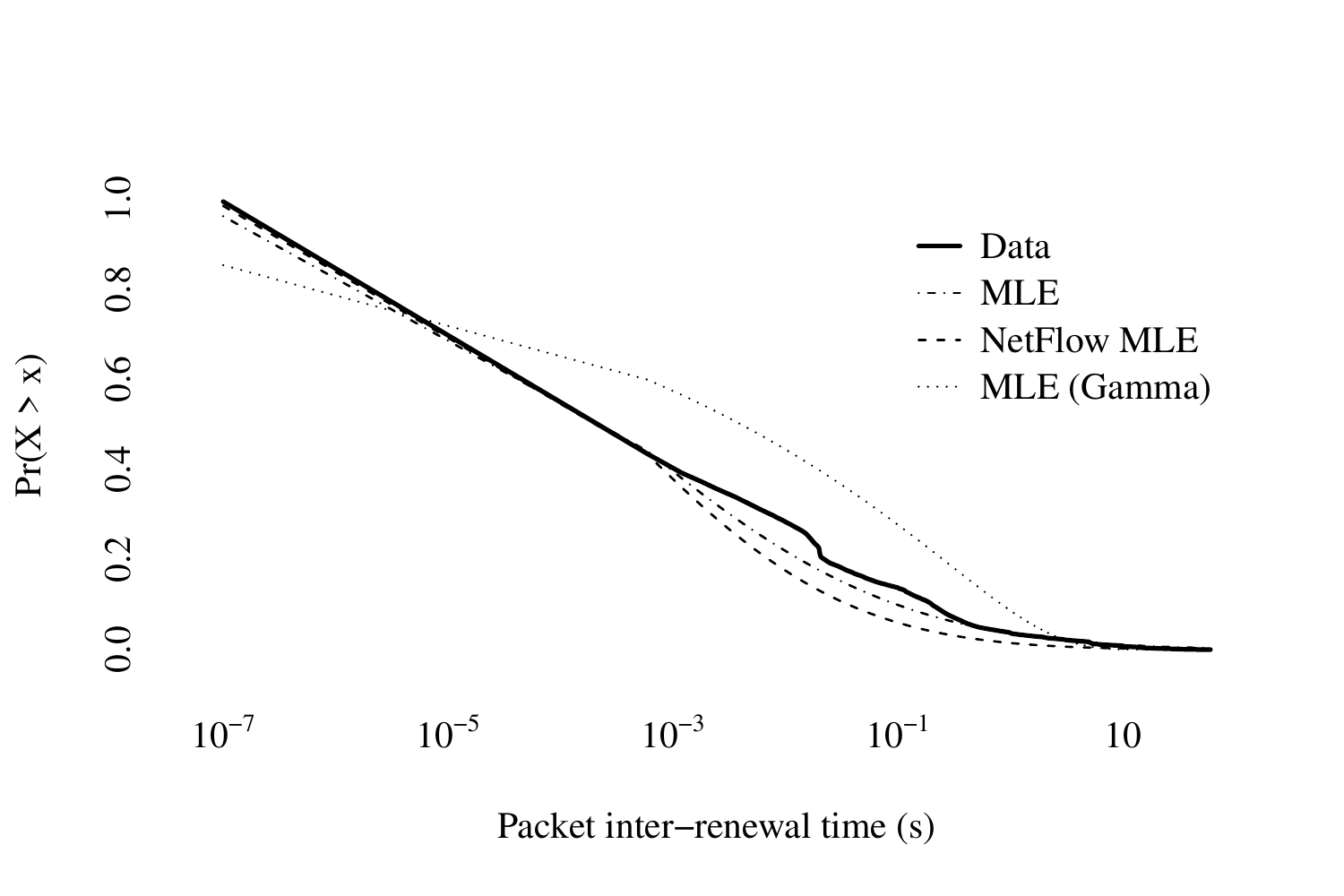}} \\
    \subfigure[$q = 0.001$]{\includegraphics[scale = 0.8]{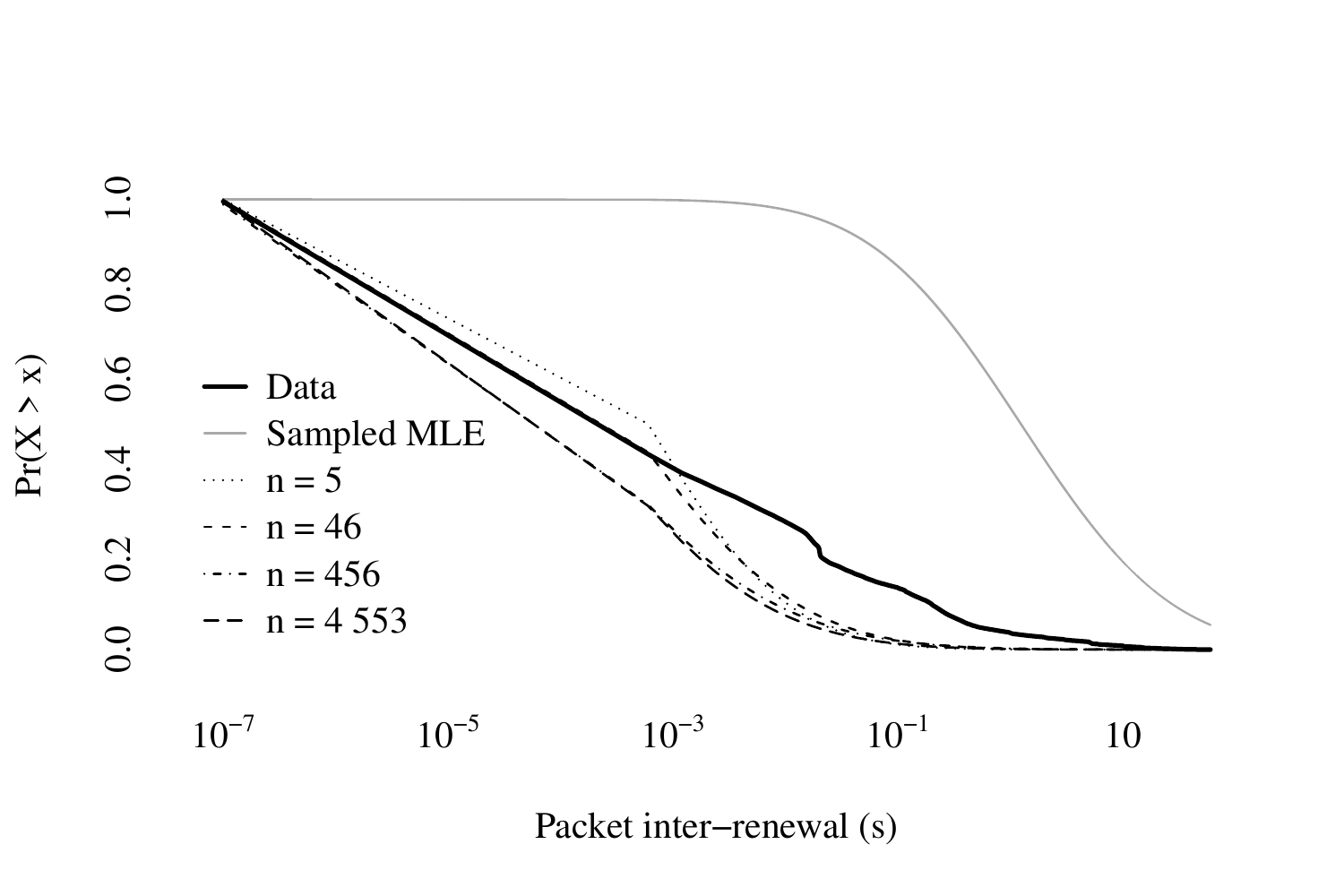}} 
    \caption{\small{
    Survival functions for the observed inter-renewals and various fitted models.
    Figure~\ref{fig:data_survival_functions}(a) supplies the survival function under complete packet retention, i.e.~$q = 1$.
    Figure~\ref{fig:data_survival_functions}(b) depicts the fitted models when each packet in the network is retained through a Bernoulli trial with probability $q = 0.001$.
    The thick black line indicates the survival function of the observed data and is identical in both (a) and (b).
    The dotted line in (a) is obtained by fitting the \textsc{mle} for a Gamma model to the complete set of inter-renewals.
    All other supplied survival functions are fitted to the Log-Normal distribution.
    The solid grey line in (b) indicates the standard \textsc{mle} which is na\"{i}vely applied to the sampled inter-renewals.
    The dotted and dashed lines in (b) correspond to the fitted models from the sampled NetFlow \textsc{mle} obtained from Table~\ref{tab:real_data} for the given number of datapoints $n$, as indicated.
    }}
    \label{fig:data_survival_functions}
\end{figure}

\section{Discussion}
We have introduced a novel method for parametric, likelihood-based inference of network packet data when the utilised data are (possibly) thinned and aggregated.
The ability to jointly handle packet thinning and NetFlow aggregation, within the likelihood framework (with all its inferential benefits) is a great advantage over existing methods for analysing flow data, which can only account for one of these processes.
The maximum likelihood estimators themselves are consistent (with increasing numbers of NetFlows), and we have derived bounds on the number of (thinned) NetFlows needed to attain a given degree of accuracy.
As a result, the NetFlow likelihoods offer a practical and flexible tool for inference on very large session datasets.

The potentially large computational cost is the price to pay for this framework.
Indeed, computing the NetFlow likelihoods requires a convolution of densities which can be particularly complex when the model is not closed under convolutions.
However, we have shown that, even on real data, closed-form approximations can be utilised without severe consequences.
Optimising a likelihood is slower than methods that rely on simple arithmetic computation (such as the moments-based estimators of \citet{Hohn03} and \citet{Antunes18}).
Computation also increases with higher degrees of packet thinning (i.e.~low $q$), or with a high frequency of elephant flows, since the NetFlow likelihood function needs to consider (and integrate out) all possible latent flows which could have produced the observed, thinned flow.

Despite these limitations, the NetFlow likelihood estimators provide an effective method for network analysis when the data we have to work with is less than ideal: both heavily summarised, and heavily subsampled.

\section*{Acknowledgements}
This research is supported by the Australian Research Council through the Australian Centre
of Excellence for Mathematical and Statistical Frontiers (ACEMS; CE140100049), and the
Discovery Project scheme (FT170100079).

Support for \textsc{caida}'s Internet Traces is provided by the National Science Foundation, the \textsc{us} Department of Homeland Security, and \textsc{caida} Members.

This research includes computations using the computational cluster Katana supported by Research Technology Services at \textsc{unsw} Sydney.


\bibliographystyle{plainnat}
\bibliography{references}

\appendix
\section{Proofs} \label{app:proofs}
\subsection{Proof to Proposition~\ref{prop:netflow_likelihood}} \label{proof:netflow_likelihood}
\begin{proof}
The conditional density of an observed Netflow $s$ is
\begin{equation} \label{eq:symbol_conditional_density}
    f_{S\vert\bm{X}}{\left(s\vert\bm{x}\right)} = \delta_{s_{f}}{\left(x_{1}\right)}\,\delta_{s_{d}}{\left(1_{\{m > 0\}}\cdot\sum^{m + 1}_{k = 2} x_{k}\right)}\,\delta_{m + 1}{\left(\left|\bm{x}\right|\right)},
\end{equation}

\noindent where $\delta_{a}{\left(\cdot\right)}$ is the Dirac measure centred at $a$. 
Equation \eqref{eq:symbol_conditional_density} simply delineates that the time between flows, flow duration, and flow size correspond exactly with the same quantities obtained from the observed sequence of inter-renewal times $\bm{x}$.
The internal variation of the Netflow is defined by the variability of the non-leading inter-renewals $x_{2},\ldots,x_{m+1}$.

Let $p_{M}{\left(m + 1;\nu\right)}$ be the mass function for the random flow size $M + 1$.
The conditional density of the sequence of inter-renewals $\bm{x}$ is
\begin{equation} \label{eq:interrenewal_times_conditional_density}
    g_{\bm{X}\vert M}{\left(\bm{x}\vert m;\bm{\theta}\right)} = \prod^{m + 1}_{i = 1} g_{i}{\left(x_{i};\theta_{i}\right)}.
\end{equation}
\noindent We then have the joint density 
\begin{equation} \label{eq:flow_interrenewal_joint_density}
g_{\bm{X}}{\left(\bm{x};\bm{\theta},\nu\right)} = \prod^{m + 1}_{i = 1} g_{i}{\left(x_{i};\theta_{i}\right)}\,p_{M}{\left(m + 1;\nu\right)}.
\end{equation}
\noindent Substituting equations \eqref{eq:symbol_conditional_density} and \eqref{eq:flow_interrenewal_joint_density} into \eqref{eq:generative_symbolic_likelihood} and computing the integral then yields the result.
\end{proof}

\subsection{Proof to Proposition~\ref{prop:sampled_netflow_likelihood}} \label{proof:sampled_netflow_likelihood}
\begin{proof}
The conditional density for the sub-sampled NetFlow symbol $\widetilde{s}$ given the sub-sampled inter-renewal sequence $\widetilde{\bm{x}}$ is
\begin{equation} \label{eq:sampled_netflow_conditional_density}
    f_{\widetilde{S}\vert\widetilde{\bm{X}}}\left(\widetilde{s}\vert\widetilde{\bm{x}}\right) = \delta_{\widetilde{s}_{f}}{\left(\widetilde{x}_{1}\right)}\,\delta_{\widetilde{s}_{d}}{\left(1_{\left\{\widetilde{m} > 1\right\}}\cdot\sum^{\widetilde{m}}_{i = 2} \widetilde{x}_{i}\right)}\,\delta_{\widetilde{m}}{\left(\lvert\widetilde{\bm{x}}\rvert\right)}.
\end{equation}

\noindent As with \eqref{eq:symbol_conditional_density}, the conditional density $f_{\widetilde{S}\vert\widetilde{\bm{X}}}$ delineates that the observed sub-sampled NetFlow symbol $\widetilde{s}$ matches the restrictions of the sub-sampled inter-renewal times $\widetilde{\bm{x}}$.

We can construct a sub-sampled equivalent of \eqref{eq:netflow_likelihood} by integrating over the space of sub-sampled inter-renewal sequences, so that
\begin{equation} \label{eq:sampled_netflow_likelihood_integral}
    \mathcal{L}_{\widetilde{S}}{\left(\widetilde{s};\bm{\theta}\right)} \propto \int_{\widetilde{\mathcal{X}}} f_{\widetilde{S}\vert\widetilde{\bm{X}}}{\left(\widetilde{s}\vert\widetilde{\bm{x}}\right)}\,g_{\widetilde{\bm{X}}}{\left(\widetilde{\bm{x}};\bm{\theta}\right)}\,\mathrm{d}\widetilde{\bm{x}},
\end{equation}
\noindent where the region of integration $\widetilde{\mathcal{X}}$ is the set of sampled inter-renewal sequences of length $\widetilde{m}$.
Although computation of the integrals for both the NetFlow and sub-sampled NetFlow likelihoods appears identical, there is a material difference in the marginal densities of the original and sampled inter-renewal sequences $g_{\bm{X}}$ and $g_{\widetilde{\bm{X}}}$, respectively.

We obtain the marginal density for the sampled inter-renewal sequence by integrating over the joint density of its generating components; the original flow size, original inter-renewal sequence, and packet sampling procedure. 

The number of sampled packets $\widetilde{M}$ is conditionally Binomially distributed since each packet arrival was independently Bernoulli sampled, that is, $\widetilde{M}\vert M + 1\sim \text{Bin}\left(M + 1, q\right)$. 
It follows that the conditional probability of obtaining a particular sampled inter-renewal sequence $\widetilde{\bm{x}}_{i}$ of length $\widetilde{m}$ follows a discrete uniform distribution, that is,
\begin{equation} \label{eq:particular_sampled_interrenewals_conditional_density}
  p_{\widetilde{\bm{X}}_{i}\vert M,\widetilde{M}}{\left(\widetilde{\bm{x}}_{i}\vert m, \widetilde{m}\right)} = \binom{m + 1}{\widetilde{m}}^{-1},    
\end{equation}
\noindent for $i=1,\ldots,\binom{m + 1}{\widetilde{m}}$.
An application of Bayes' theorem yields the conditional mass function of the original flow size
\begin{equation} \label{eq:flow_size_conditional_density}
\begin{aligned}
    p_{M\vert\widetilde{M}}{\left(m + 1\vert\widetilde{m};\nu,q\right)} &= \frac{p_{M,\widetilde{M}}{\left(m,\widetilde{m};\nu,q\right)}}{p_{\widetilde{M}}{\left(\widetilde{m};\nu,q\right)}}1_{\left\{m + 1\geq\widetilde{m}\right\}} \\
                                                                        &= \frac{p_{\widetilde{M}\vert M}{\left(\widetilde{m}\vert m;q\right)}\,p_{M}{\left(m + 1;\nu\right)}}{p_{\widetilde{M}}{\left(\widetilde{m};\nu,q\right)}}1_{\left\{m + 1\geq\widetilde{m}\right\}}.
\end{aligned}
\end{equation}
\noindent The marginal density of the sampled flow size is
\begin{equation} \label{eq:sampled_flow_size_marginal_density}
  p_{\widetilde{M}}{\left(\widetilde{m};\nu, q\right)} = \sum_{m + 1 \geq \widetilde{m}} \binom{m + 1}{\widetilde{m}}q^{\widetilde{m}}(1 - q)^{m + 1 - \widetilde{m}}\,p_{M}{\left(m + 1;\nu\right)}.
\end{equation}
\noindent Taking the product of the densities \eqref{eq:interrenewal_times_conditional_density}, \eqref{eq:particular_sampled_interrenewals_conditional_density}, 
\eqref{eq:flow_size_conditional_density}, and \eqref{eq:sampled_flow_size_marginal_density} yields the joint density 
\begin{displaymath}
f_{\widetilde{\bm{X}}_{i},\bm{X},M,\widetilde{M}}{\left(\widetilde{\bm{x}}_{i}, \bm{x}, m, \widetilde{m};\theta,\nu,q\right)}.
\end{displaymath}
\noindent The marginal density of the sampled inter-renewal sequence of length $\widetilde{m}$ is given by the integral
\begin{equation} \label{eq:sampled_interrewal_times_marginal_density}
    g_{\widetilde{\bm{X}}}\left(\widetilde{\bm{x}};\bm{\theta}\right) = \int\cdots\int_{\Omega} f_{\widetilde{\bm{X}}_{i},\bm{X},M,\widetilde{M}}{\left(\widetilde{\bm{x}}_{i}, \bm{x}, m, \widetilde{m};\theta,\nu,q\right)}\;\mathrm{d}\mu,
\end{equation}
\noindent where $\mu$ is an appropriately defined product measure composed from the Lebesgue and counting measures, and $\Omega$ is the Cartesian product of the conditional sample spaces for $M$, $\bm{X}$ and $\widetilde{\bm{X}}_{i}$.
Explicitly, we write 
\begin{displaymath}
    \Omega = \N\times\R^{(m + 1)}_{+}\times\left\{1,\ldots,\binom{m + 1}{\widetilde{m}}\right\}.
\end{displaymath}
\noindent Substituting \eqref{eq:sampled_netflow_conditional_density} and \eqref{eq:sampled_interrewal_times_marginal_density} into \eqref{eq:sampled_netflow_likelihood_integral} and computing the integral with respect to $\delta_{\widetilde{m}}{\left(\vert\widetilde{\bm{x}}\vert\right)}$ yields the intermediate expression
\begin{equation} \label{eq:intermediate_sampled_netflow_likelihood}
\mathcal{L}_{\widetilde{S}}{\left(\widetilde{s};\bm{\theta}\right)} \propto \sum_{m + 1 \geq \widetilde{m}} \Psi_{m}{\left(\widetilde{m};\nu,q\right)}\int_{\mathcal{X}}g_{\bm{X}\vert M}{\left(\bm{x};\bm{\theta}\right)}\,\delta_{\widetilde{s}}{\left(\widetilde{\bm{x}}\right)}\;\mathrm{d}\mu,
\end{equation}

\noindent where 
\begin{displaymath}
\begin{aligned}
    \Psi_{m}{\left(\widetilde{m};\nu,q\right)}              &= p_{M}{\left(m + 1;\nu\right)}\,\tau_{m}{\left(\widetilde{m};q\right)}, \\
    \delta_{\widetilde{s}}{\left(\widetilde{\bm{x}}\right)} &= \delta_{\widetilde{s}_{f}}\left(\widetilde{x}_{1}\right)\,\delta_{\widetilde{s}_{d}}{\left(1_{\left\{\widetilde{m} > 1\right\}}\cdot\sum^{\widetilde{m}}_{i = 2} \widetilde{x}_{i}\right)}, 
\end{aligned}
\end{displaymath}
\noindent and the conditional density $g_{\bm{X}\vert M}$ is defined in \eqref{eq:interrenewal_times_conditional_density}.
Computing the integral in \eqref{eq:intermediate_sampled_netflow_likelihood} will then conclude the proof.
Note that the region of integration can be written as 
\begin{displaymath}
    \mathcal{X} = \R^{(m + 1)}_{+} = \R^{J}_{+}\times\R^{K}_{+}\times\R^{(m + 1 - (J + K))}_{+},
\end{displaymath}
\noindent where $J$ and $K$ are integral-valued random variables respectively denoting the location of the first sampled packet and the number of inter-renewals between the first and the last sampled packets.
Hence, we can compute the marginal integral with respect to the random variables $J$ and $K$. 
By splitting the product $\prod^{m + 1}_{i = 1} g_{i}{\left(x_{i};\theta_{i}\right)}$ according to $J$ and $K$ and integrating with respect to $\delta_{\widetilde{s}}{\left(\widetilde{\bm{x}}\right)}$, we obtain the integral
\begin{displaymath}
\int_{\mathcal{J,K}} \mathcal{G}_{j,k}{\left(\widetilde{s};\bm{\theta}\right)}\,p_{J, K\vert M,\widetilde{M}}{\left(j,k\vert m,\widetilde{m}\right)}\;\mathrm{d}\mu,
\end{displaymath}

\noindent where the joint density 
\begin{displaymath}
\begin{aligned}
    p_{J,K\vert M,\widetilde{M}} &= p_{J\vert M,\widetilde{M}}(j\vert m,\widetilde{m})\,p_{K\vert M,\widetilde{M},J}(k\vert m,\widetilde{m},j) \\
                                 &= \frac{\binom{m + 1 - j}{\widetilde{m} - 2}}{\binom{m + 1}{\widetilde{m}}}\,\frac{\binom{k - 1}{\widetilde{m} - 2}}{\binom{m + 1 - j}{\widetilde{m} - 1}} \\
                                 &= \frac{\binom{k - 2}{\widetilde{m} - 1}}{\binom{m + 1}{\widetilde{m}}}.
\end{aligned}
\end{displaymath}
\noindent The integral then simplifies to 
\begin{displaymath}
    \sum^{m + 2 - \widetilde{m}}_{j = 1}\sum^{m + 1 - j}_{k = \widetilde{m} - 1}\upsilon_{m,k}{\left(\widetilde{m}\right)}\mathcal{G}{\left(\widetilde{s};\bm{\theta}\right)},
\end{displaymath}
\noindent which yields \eqref{eq:sampled_netflow_likelihood} when substituted into the intermediate likelihood \eqref{eq:intermediate_sampled_netflow_likelihood}.
\end{proof}

\subsection{Proof to Proposition~\ref{prop:netflow_estimator_consistency}} \label{proof:netflow_estimator_consistency}
\begin{proof}
We first show that the fully sampled NetFlow estimator $\hat{\theta}_{S}$ is consistent.
We then define conditions which give consistency for the thinned NetFlow estimator $\hat{\theta}_{\tilde{S}}$ when applying an identical approach.
Firstly, we assume that $\Theta$ is locally compact and we denote a general compact neighbourhood by $\overline{\Theta}$.

\begin{definition}
Let $\Lambda(s;\theta) := \E_{\theta_{0}}{\left[\ell_{1}(s;\theta)\right]}$ be the $\theta$-limit of \eqref{eq:netflow_log_likelihood}.
\end{definition}

\begin{lemma} \label{lem:pointwise_convergence}
The sequence $\left\{\ell_{n}(\bm{s};\theta)\right\}_{n\geq1}$ converges to $\Lambda(s;\theta)$ pointwise by the law of large numbers.
\end{lemma}

\begin{definition}[Identifiability] \label{def:identifiability}
A model $f(x,\theta)$ is \emph{identifiable} if $f(x,\theta_{1}) = f(x,\theta_{2})$ implies that $\theta_{1} = \theta_{2}$ for all $\theta_{1},\theta_{2}\in\overline{\Theta}$.
In other words, the model $f(x,\theta)$ is one-to-one in $\overline{\Theta}$.
\end{definition}

\begin{lemma}
Suppose that the supplied model $g_{X}$ is identifiable.
Then $\Lambda(s,\theta)$ is uniquely maximised at $\theta_{0}$.
\end{lemma}

\begin{proof}
The inequality 
\begin{displaymath}
    \E_{\theta_0}{\left[-\log{\left(\frac{g^{*(m)}_{X}{(x;\theta})}{g^{*(m)}_{X}{(x;\theta_{0}})}\right)}\right]} > -\log{\left(\E_{\theta_0}{\left[\frac{g^{*(m)}_{X}{(x;\theta})}{g^{*(m)}_{X}{(x;\theta_{0}})}\right]}\right)}
\end{displaymath}
\noindent arises from an application of Jensen's inequality to the strictly convex function $-\log(x)$.
Noting that the right hand side is zero, it follows that
\begin{displaymath}
    \E_{\theta_0}{\left[g^{*(m)}_{X}{\left(x;\theta_{0}\right)}\right]} > \E_{\theta_0}{\left[g^{*(m)}_{X}{\left(x;\theta\right)}\right]} 
\end{displaymath}
\noindent for all $\theta\neq\theta_{0}$.
Hence, $\Lambda(s;\theta)$ is uniquely maximised at $\theta_{0}$.
\end{proof}

\begin{definition} [Stochastic equicontinuity \citep{Newey91}]
A sequence of real-valued random functions $\left\{f_{n}(x;\theta)\right\}_{n\geq1}$ is \emph{stochastically equicontinuous} in $\theta$ if, for all positive $\varepsilon$ and $\eta$, there exists a positive $\delta$ such that
\begin{displaymath}
    \underset{n\rightarrow\infty}{\lim\sup}\,\Pr\left(\sup_{\vartheta\in\overline{\Theta}}\,\sup_{\theta\in B(\vartheta, \delta)}\left\vert f_{n}(x;\vartheta) - f_{n}(x;\theta)\right\vert > \varepsilon\right) < \eta,
\end{displaymath}

\noindent where $B(\vartheta, \delta)$ is the open metric ball centred at $\vartheta$ with radius $\delta$.
\end{definition}

\begin{lemma}
The sequences $\left\{\ell_{n}(\bm{s};\theta)\right\}^{\infty}_{n=1}$ and $\left\{\Lambda(s;\theta)\right\}^{\infty}_{n=1}$ are stochastically equicontinuous.
\end{lemma}

\begin{proof}
 Recall that $\overline{\Theta}$ is compact and that the supplied model $g_{X}$ is continuous in $\Theta$.
 Hence, $\left\{\Lambda(s;\theta)\right\}_{n\geq1}$ is stochastically equicontinuous since $\Lambda(s;\theta)$ is constant in $n$ and uniformly continuous in $\overline{\Theta}$.
  
We now prove that the sequence $\left\{\ell_{n}\left(\bm{s};\theta\right)\right\}_{n\geq1}$ is stochastically equicontinuous. 
Firstly, let 
\begin{displaymath}
    \gamma{\left(s_{i};\vartheta,\theta\right)} = \log{\left(g^{*(m_{i})}{\left(s_{d_{i}};\vartheta\right)}\right)} - \log{\left(g^{*(m_{i})}{\left(s_{d_{i}};\theta\right)}\right)}.
\end{displaymath}

\noindent Then, applying the triangle inequality, Markov's inequality, and linearity of the expectation, for some positive $\varepsilon$ and $\delta$, we have that
\begin{equation} \label{eq:stochastic_equicontinuity_netflow_likelihood}
\begin{aligned}
    &\underset{n\to\infty}{\lim\sup}\,\Pr{\left(\sup_{\vartheta\in\overline{\Theta}}\sup_{\theta\in B(\vartheta,\delta)}\left\vert\ell_{n}{\left(\bm{s};\vartheta\right)} - \ell_{n}{\left(\bm{s};\theta\right)}\right\vert > \varepsilon\right)} \\
    &\quad= \underset{n\to\infty}{\lim\sup}\,\Pr{\left(\sup_{\vartheta\in\overline{\Theta}}\sup_{\theta\in B(\vartheta,\delta)}\left\vert\frac{1}{n}\sum^{n}_{i = 1}\gamma{\left(s_{i};\vartheta,\theta\right)}\right\vert > \varepsilon\right)} \\
    &\quad\leq \underset{n\to\infty}{\lim\sup}\,\Pr{\left(\sum^{n}_{i = 1}\sup_{\vartheta\in\overline{\Theta}}\sup_{\theta\in B(\vartheta,\delta)}\left\vert\gamma{\left(s_{i};\vartheta,\theta\right)}\right\vert > n\varepsilon\right)} \\
    &\quad\leq \underset{n\to\infty}{\lim\sup}\,\frac{1}{n\varepsilon}\E{\left[\sum^{n}_{i = 1}\sup_{\vartheta\in\overline{\Theta}}\sup_{\theta\in B(\vartheta,\delta)}\left\vert\gamma{\left(s_{i};\vartheta,\theta\right)}\right\vert\right]} \\
    &\quad\leq \frac{1}{\varepsilon}\E{\left[\sup_{\vartheta\in\overline{\Theta}}\sup_{\theta\in B(\vartheta,\delta)}\left\vert\gamma{\left(s_{j};\vartheta,\theta\right)}\right\vert\right]},
\end{aligned}
\end{equation}
\noindent where 
\begin{displaymath}
    j = \underset{i=1,\ldots,n}{\arg\max}\,\E{\left[\sup_{\vartheta\in\overline{\Theta}}\sup_{\theta\in B(\vartheta,\delta)}\left\vert\gamma{\left(s_{i};\vartheta,\theta\right)}\right\vert\right]}.
\end{displaymath}

\noindent For some positive $\eta$, we can choose $\delta^{\prime} = \delta{\left(\varepsilon, \eta\right)}$ such that 
\begin{displaymath}
    \sup_{\theta\in B{\left(\vartheta,\delta^{\prime}\right)}}\left\vert \log{\left(g^{*(m_{j})}{\left(\widetilde{s}_{d_j};\vartheta\right)}\right)} - \log{\left(g^{*(m_{j})}{\left(\widetilde{s}_{d_j};\theta\right)}\right)}\right\vert < \eta\varepsilon,
\end{displaymath}
\noindent since $\log{\left(g^{*(k)}{\left(x;\theta\right)}\right)}$ is also uniformly continuous in $\overline{\Theta}$.
Naturally, the inequality is also satisfied for all $i\neq j$.
Hence, we can strictly bound \eqref{eq:stochastic_equicontinuity_netflow_likelihood} from above by $\eta$.
It follows that the sequence $\left\{\ell_{n}\left(\bm{s};\theta\right)\right\}_{n\geq1}$ is stochastically equicontinuous.
\end{proof}

By an application of Theorem~2.1 of \citep{Newey91}, we surmise that $\ell_{n}$ converges to $\Lambda$ uniformly in probability.
An application of the \emph{continuous mapping theorem} with respect to the argmax functional then yields consistency of $\hat{\theta}_{S}$ for $\theta_0$.

The proof for the sampled NetFlow estimator $\hat{\theta}_{\widetilde{S}}$ is identical if the cardinality of the flow sizes is bounded.
However, if the cardinality is $\aleph_{0}$, we require the series \eqref{eq:series_convergence} to be jointly uniformly convergent so that the NetFlow log-likelihood is also uniformly continuous in $\theta$.
We can the pursue the same method of proof with the stated condition.
\end{proof}

\subsection{Proof to Proposition~\ref{prop:netflow_bounds}} \label{proof:netflow_bounds}
\begin{proof}
Firstly, let $k$ be sufficiently large.
Note that $k = 1$ may suffice for suitable flow size structures.

We can stochastically bound the variance between the \textsc{mle} and the NetFlow estimator by considering the probability
\begin{equation} \label{eq:probability_volumes}
    \Pr{\left(\left\vert\frac{\vert\Upsilon\vert^{1/2}}{\vert\Sigma\vert^{1/2}} - 1\right\vert \geq \varepsilon\right)} < \eta,
\end{equation}
\noindent since $\overline{M}$ is random, for some positive $\varepsilon$ and $\eta\in{(0,1)}$.
Let $r = n / k$ and $R = \vert I\vert / \vert H\vert$, respectively denoting the asymptotic ratio of NetFlows to complete flows and the ratio of Fisher information. 
Then, squaring the expression within the probability, substituting for $\Sigma$ and $\Upsilon$, and rearranging yields
\begin{displaymath}
    r^{-d}R^{-1}\overline{M}^{d} - 2r^{-d/2}R^{-1/2}\overline{M}^{d / 2} + 1 - \varepsilon^{2} \geq 0.
\end{displaymath}
\noindent The expression above is a positive quadratic with respect to $\overline{M}^{d / 2}$ and has roots
\begin{displaymath}
    x_{\pm} = r^{d / 2}R^{1/2}(1\pm\varepsilon). 
\end{displaymath}

\noindent Hence, the probability in \eqref{eq:probability_volumes} can be identically expressed as
\begin{equation} \label{eq:probability_volume_rewritten}
    \Pr{\left(\overline{M}\geq rR^{1/d}(1+\varepsilon)^{2/d}\right)} + \Pr{\left(\overline{M}\leq rR^{1/d}(1-\varepsilon)^{2/d}\right)} < \eta.
\end{equation}

\noindent Bounding each of the summands above by $\eta / 2$ will satisfy the expression.
Applying Chernoff's bound to the first summand, bounding the limit by $\eta / 2$, and rearranging for $r$ yields
\begin{displaymath}
    r > \sqrt[d]{\frac{R^{-1}}{(1 + \varepsilon)^{2}}}\log{\left(\frac{2}{\eta}\E{\left[\mathrm{e}^{\overline{M}}\right]}\right)}.
\end{displaymath}

\noindent Performing a similar set of operations to the latter summand yields the upper bound
\begin{displaymath}
    r < -\sqrt[d]{\frac{R^{-1}}{(1 - \varepsilon)^{2}}}\log{\left(\frac{2}{\eta}\E{\left[\mathrm{e}^{-\overline{M}}\right]}\right)}.
\end{displaymath}

\noindent We require that
\begin{displaymath}
    \E{\left[\mathrm{e}^{\overline{M}}\right]}\,\E{\left[\mathrm{e}^{-\overline{M}}\right]}^{A} < \left(\frac{\eta}{2}\right)^{A + 1}
\end{displaymath}

\noindent if we wish to jointly consider the upper and lower bounds, where $A = ((1 +\varepsilon) / (1 - \varepsilon))^{2/d}$.
Note that $A\approx1$ for particularly small $\varepsilon$.
\end{proof}

\subsection{Proof to Proposition~\ref{lem:identical_inference}} \label{proof:identical_inference}
\begin{proof}
Recalling the densities in \eqref{eq:natural_exponential_density} and \eqref{eq:natural_exponential_family_convolution_density}, for a sequence of inter-renewals $\bm{x}$ and its associated NetFlow $s = (s_{f}, s_{d}, m + 1)$, we have the respective usual and NetFlow likelihoods
\begin{displaymath}
\begin{aligned}
    \mathcal{L}{\left(\bm{x};\theta\right)} &\propto \exp{\left(\theta\sum^{m + 1}_{i = 1}x_{i} - (m + 1)A(\theta)\right)}\text{ and} \\
    \mathcal{L}_{S}{\left(s;\theta\right)}  &\propto \exp{\left(\theta\left(s_{f} + s_{d}\right) - (m + 1)A(\theta)\right)}.
\end{aligned}
\end{displaymath}

\noindent It is clear that $\mathcal{L}$ and $\mathcal{L}_{S}$ are identical since $\sum^{m+1}_{k=1}x_{k} = s_{f} + s_{d}$.
\end{proof}

\section{Method-of-moments} \label{app:simulations}
The moments-based rate estimator $\hat{\beta} = \hat{\alpha}\sum^{n}_{i = 1}w_{i}\varrho_{i}$ notably failed to converge for the synthetic network established in Example~\ref{ex:moments}.
This peculiarity arose from distributional properties of the packet intensity $\varrho$.
Recall that $\varrho$ has Inverse-Gamma distribution with parameters $\alpha'$ and $\beta$.
The expected value of $\varrho$ is not defined when $\alpha' < 1$.
This is observed in Example~\ref{ex:moments} whenever the flow size $M + 1 \leq 2$.
However, we can ensure that $\alpha' > 1$ if the shape parameter for the Gamma distributed inter-renewals $\alpha > 1$.
Alternatively, we can compute $\varrho$ only for flows with size $M > 1 / \alpha$ since the shape parameter of the convolved Gamma inter-renewals will be larger than one.

We perform identical comparative analyses to Example~\ref{ex:moments} but for the aforementioned conditions.
We maintain the same properties for the flow sizes but now consider two synthetic networks with parameters
\begin{enumerate}
    \item $(\alpha_{1}, \beta_{1}) = (1.2, 526.32)$,
    \item $(\alpha_{2}, \beta_{2}) = (0.6, 526.32)$ with $M + 1\geq 3$.
\end{enumerate}

\noindent Networks (1) and (2) respectively satisfy conditions where $\alpha > 1$ and $M > 1 / \alpha$.
Tables \ref{tab:moments_scenario_2} and \ref{tab:metadata_moments_scenario_2} respectively present the mean point estimate (and standard errors) and metadata for network (1).
Tables \ref{tab:moments_scenario_3} and \ref{tab:metadata_moments_scenario_3} similarly present the mean point estimate (and standard errors) and metadata for network (2). 
The total volume of information generated in simulating network (2) for sessions with $n = 10^{6}$ flows exceeded the local memory, and so the moment estimates were not computed. 
However, it is possible to generate an equivalent set of $10^{6}$ storable NetFlows with the same generating parameters.

Although the adjustments in networks (1) and (2) yield somewhat better estimates for for $\hat{\beta}$ and $\check{\beta}$, overestimation will persist due to the right skewness of the Inverse-Gamma distribution.
As in Example~\ref{ex:moments}, the NetFlow \textsc{mle} is preferable to the moments-based estimators.
However applications of the na\"{i}ve rate estimator $\hat{\beta}^{*}$ yields significantly better performance when utilising the complete set of inter-renewals, as expected, and is comparable to the NetFlow \textsc{mle} when restricted to flow aggregated data.

\begin{table}[t]
    \caption{\small{
    Example~\ref{ex:moments}: Mean point estimates (and standard errors) of $(\alpha,\beta)$ for network (1) from on $T = 10^{3}$ replicate synthetic sessions of size $n$. 
    True values are $(\alpha_{2},\beta_{2})=(1.2, 526.32)$.
    }}
    \label{tab:moments_scenario_2}
    \centering
    \begin{tabular}{@{}rrrrr@{}} \toprule
                        & \multicolumn{4}{c}{$n$}                                                                                                   \\ \cmidrule{2-5}
                        & $10^{0}$                     & $10^{2}$                     & $10^{4}$                     & $10^{6}$                     \\ \midrule
    \multicolumn{5}{l}{\small{\emph{Method-of-moments}}}                                                                                            \\
    $\hat{\alpha}$      & 30.67\p                      & 1.23\p                       & 1.20\p                       & 1.20\p                       \\
                        & (2.05)                       & $\left({\sim}10^{-3}\right)$ & $\left({\sim}10^{-4}\right)$ & $\left({\sim}10^{-5}\right)$ \\ \addlinespace[1ex]
    $\hat{\beta}$       & ${\sim}10^{4}$\p             & 713.70\p                     & 621.96\p                     & 589.03\p                     \\
                        & $\left({\sim}10^{3}\right)$  & (11.38)                      & (6.88)                       & (0.76)                       \\ \addlinespace[1ex]
    $\hat{\beta}^{*}$   & ${\sim}10^{4}$\p             & 539.18\p                     & 526.24\p                     & 526.31\p                     \\
                        & $\left({\sim}10^{3}\right)$  & (2.23)                       & (0.15)                       & (0.01)                       \\ \addlinespace[1ex]
    \multicolumn{5}{l}{\small{\emph{NetFlow method-of-moments}}}                                                                                    \\
    $\check{\alpha}$    & ---                          & 1.40\p                       & 1.20\p                       & 1.20\p                       \\ 
                        & ---                          & (0.01)                       & $\left({\sim}10^{-4}\right)$ & $\left({\sim}10^{-5}\right)$ \\ \addlinespace[1ex] 
    $\check{\beta}$     & ---                          & 807.93\p                     & 624.01\p                     & 589.00\p                     \\
                        & ---                          & (5.76)                       & (2.88)                       & (0.31)                       \\ \addlinespace[1ex]
    $\check{\beta}^{*}$ & ---                          & 617.16\p                     & 527.86\p                     & 526.28\p                     \\
                        & ---                          & (2.59)                       & (0.28)                       & (0.03)                       \\ \addlinespace[1ex]
    \multicolumn{5}{l}{\small{\emph{NetFlow \textsc{mle}}}}                                                                                         \\
    $\hat{\alpha}_{S}$  & ${\sim}10^{30}$\p            & 1.31\p                       & 1.20                         & 1.20\p                       \\
                        & $\left({\sim}10^{29}\right)$ & (0.01)                       & $\left({\sim}10^{-4}\right)$ & $\left({\sim}10^{-5}\right)$ \\ \addlinespace[1ex]
    $\hat{\beta}_{S}$   & ${\sim}10^{33}$\p            & 575.40\p                     & 526.48\p                     & 526.35\p                     \\
                        & $\left({\sim}10^{33}\right)$ & (4.60)                       & (0.37)                       & (0.04)                       \\ \bottomrule 
    \end{tabular}
\end{table}

\begin{table}[t]
    \caption{\small{
    Example~\ref{ex:moments}: Mean session information volume (megabytes) and computation time (milliseconds) for network (1) over various session sizes $n$.
    }}
    \label{tab:metadata_moments_scenario_2}
    \centering
    \begin{tabular}{@{}lrrrr@{}} \toprule
                              & \multicolumn{4}{c}{$n$}                                             \\ \cmidrule{2-5}                          
                              & $10^{0}$        & $10^{2}$         & $10^{4}$        & $10^{6}$     \\ \midrule
    \multicolumn{5}{l}{\small{\emph{Method-of-moments}}}                                                      \\
    Information (\textsc{mb}) & ${\sim}10^{-4}$ & 0.01             & 0.89            & 116.20       \\
    Time $(ms)$               & ${\sim}10^{-1}$ & ${\sim}10^{-2}$  & 1\ppp           & 249\ppp      \\ \addlinespace[1ex]
    \multicolumn{5}{l}{\small{\emph{NetFlow method-of-moments}}}                                       \\
    Information (\textsc{mb}) & ---             & ${\sim}10^{-3}$  & 0.09            & 9.31         \\
    Time $(ms)$               & ---             & ${\sim}10^{-2}$  & ${\sim}10^{-1}$ & 12\ppp       \\ \addlinespace[1ex]
    \multicolumn{5}{l}{\small{\emph{NetFlow \textsc{mle}}}}                                            \\
    Information (\textsc{mb}) & ${\sim}10^{-4}$ & ${\sim}10^{-3}$  & 0.09            & 9.31         \\
    Time $(ms)$               & 1               & 1                & 58\ppp          & $5\,673$\ppp \\ \bottomrule
    \end{tabular}
\end{table}

\begin{table}[t]
    \caption{\small{
    Example~\ref{ex:moments}: Mean point estimates (and standard errors) of $(\alpha,\beta)$ for network (2) from $T = 10^{3}$ replicate synthetic sessions of size $n$. 
    True values are $(\alpha_{3},\beta_{3})=(0.6, 526.32)$ with truncated flow sizes $M+1\geq 3$. 
    ---$^a$ indicates that the size of the dataset was too large to be read into memory. 
    }}
    \label{tab:moments_scenario_3}
    \centering
    \begin{tabular}{@{}rrrrr@{}} \toprule
                        & \multicolumn{4}{c}{$n$}                                                                                                   \\ \cmidrule{2-5}
                        & $10^{0}$                     & $10^{2}$                     & $10^{4}$                     & $10^{6}$                     \\ \midrule
    \multicolumn{5}{l}{\small{\emph{Method-of-moments}}}                                                                                            \\
    $\hat{\alpha}$      & 430.34\p                     & 0.60\p                       & 0.60\p                       & ---$^{a}$                    \\
                        & (420.24)                     & $\left({\sim}10^{-3}\right)$ & $\left({\sim}10^{-4}\right)$ & ---$^{a}$                    \\ \addlinespace[1ex]
    $\hat{\beta}$       & ${\sim}10^{5}$\p             & 677.72\p                     & 613.49\p                     & ---$^{a}$                    \\
                        & $\left({\sim}10^{5}\right)$  & (7.63)                       & (2.12)                       & ---$^{a}$                    \\ \addlinespace[1ex]
    $\hat{\beta}^{*}$         & ${\sim}10^{5}$\p             & 528.65\p                     & 526.31\p                     & ---$^{a}$              \\
                        & $\left({\sim}10^{5}\right)$  & (1.09)                       & (0.08)                       & ---$^{a}$                    \\ \addlinespace[1ex]
    \multicolumn{5}{l}{\small{\emph{NetFlow method-of-moments}}}                                                                                    \\
    $\check{\alpha}$    & ---                          & 0.63\p                       & 0.60\p                       & 0.60\p                       \\ 
                        & ---                          & $\left({\sim}10^{-3}\right)$ & $\left({\sim}10^{-4}\right)$ & $\left({\sim}10^{-4}\right)$ \\ \addlinespace[1ex] 
    $\check{\beta}$     & ---                          & 709.97\p                     & 613.91\p                     & 590.06\p                     \\
                        & ---                          & (10.85)                      & (2.21)                       & (0.76)                       \\ \addlinespace[1ex]
    $\check{\beta}^{*}$ & ---                          & 553.98\p                     & 526.65\p                     & 526.33\p                     \\
                        & ---                          & (3.80)                       & (0.40)                       & (0.04)                       \\ \addlinespace[1ex]
    \multicolumn{5}{l}{\small{\emph{NetFlow \textsc{mle}}}}                                                                                         \\
    $\hat{\alpha}_{S}$  & ${\sim}10^{30}$\p            & 0.61\p                       & 0.60\p                       & 0.60\p                       \\
                        & $\left({\sim}10^{29}\right)$ & $\left({\sim}10^{-3}\right)$ & $\left({\sim}10^{-4}\right)$ & $\left({\sim}10^{-4}\right)$ \\ \addlinespace[1ex]
    $\hat{\beta}_{S}$   & ${\sim}10^{33}$\p            & 540.03\p                     & 526.50\p                     & 526.33\p                     \\
                        & $\left({\sim}10^{33}\right)$ & (2.39)                       & (0.23)                       & (0.03)                       \\ \bottomrule 
    \end{tabular}
\end{table}

\begin{table}[t]
    \caption{\small{
    Example~\ref{ex:moments}: Mean session information volume (megabytes) and computation time (milliseconds) for network (2) over various session sizes $n$. 
    ---$^a$ indicates that the size of the dataset was too large to be read into memory.
    }}
    \label{tab:metadata_moments_scenario_3}
    \centering
    \begin{tabular}{@{}lrrrr@{}} \toprule
                              & \multicolumn{4}{c}{$n$}                                                \\ \cmidrule{2-5}                          
                              & $10^{0}$        & $10^{2}$         & $10^{4}$        & $10^{6}$        \\ \midrule
    \multicolumn{5}{l}{\small{\emph{Method-of-moments}}}                                               \\
    Information (\textsc{mb}) & ${\sim}10^{-4}$ & 0.04             & 3.76            & ${\sim}450$\ppp \\
    Time $(ms)$               & ${\sim}10^{-2}$ & ${\sim}10^{-1}$  & 6\ppp           & ---$^{a}$       \\ \addlinespace[1ex]
    \multicolumn{5}{l}{\small{\emph{NetFlow method-of-moments}}}                                       \\
    Information (\textsc{mb}) & ---             & ${\sim}10^{-3}$  & 0.24            & 24.00           \\
    Time $(ms)$               & ---             & ${\sim}10^{-2}$  & ${\sim}10^{-1}$ & 65\ppp          \\ \addlinespace[1ex]
    \multicolumn{5}{l}{\small{\emph{NetFlow \textsc{mle}}}}                                            \\
    Information (\textsc{mb}) & ${\sim}10^{-4}$ & ${\sim}10^{-3}$  & 0.24            & 24.00           \\
    Time $(ms)$               & 3              & 2                 & 126\ppp         & $20\,240$\ppp   \\ \bottomrule
    \end{tabular}
\end{table}

\end{document}